\documentclass[11pt]{article}

\usepackage{array,fullpage,multirow}
\usepackage{amssymb,amsmath,amsthm,sectsty,url,nicefrac,color}
\usepackage[bookmarks=true,pdfstartview=FitH,colorlinks,linkcolor=blue,filecolor=blue,citecolor=blue,urlcolor=blue]{hyperref}
\usepackage{cleveref}
\usepackage[pdftex]{pict2e}
\usepackage{aliascnt}
\usepackage[numbers]{natbib} 
\usepackage[inline]{enumitem}

\usepackage[boxed]{algorithm}
\usepackage{algpseudocode}

\numberwithin{equation}{section}

\newtheorem{theorem}{Theorem}[section]

\newaliascnt{lemma}{theorem}
\newtheorem{lemma}[lemma]{Lemma}
\aliascntresetthe{lemma}
\crefname{lemma}{Lemma}{Lemmas}

\newaliascnt{example}{theorem}

\aliascntresetthe{example}
\crefname{example}{Example}{Examples}

\newaliascnt{claim}{theorem}
\newtheorem{claim}[claim]{Claim}
\aliascntresetthe{claim}
\crefname{claim}{Claim}{Claims}

\newaliascnt{corollary}{theorem}

\aliascntresetthe{corollary}
\crefname{corollary}{Corollary}{Corollaries}

\newaliascnt{construction}{theorem}
\newtheorem{construction}[construction]{Construction}
\aliascntresetthe{construction}
\crefname{construction}{Construction}{Constructions}

\newaliascnt{question}{theorem}

\aliascntresetthe{question}
\crefname{question}{Question}{Questions}

\newaliascnt{fact}{theorem}

\aliascntresetthe{fact}
\crefname{fact}{Fact}{Facts}

\newaliascnt{proposition}{theorem}

\aliascntresetthe{proposition}
\crefname{proposition}{Proposition}{Propositions}

\newaliascnt{conjecture}{theorem}
\newtheorem{conjecture}[conjecture]{Conjecture}
\aliascntresetthe{conjecture}
\crefname{conjecture}{Conjecture}{Conjectures}

\newaliascnt{definition}{theorem}
\newtheorem{definition}[definition]{Definition}
\aliascntresetthe{definition}
\crefname{definition}{Definition}{Definitions}

\newaliascnt{remark}{theorem}

\aliascntresetthe{remark}
\crefname{remark}{Remark}{Remarks}

\newaliascnt{observation}{theorem}
\newtheorem{observation}[observation]{Observation}
\aliascntresetthe{observation}
\crefname{observation}{Observation}{Observations}

\crefname{algorithm}{Algorithm}{Algorithms}

\newaliascnt{notation}{theorem}

\aliascntresetthe{notation}
\crefname{notation}{Notation}{Notations}

\newcommand\E{\mathop{\mathbb E}}

\newcommand{\Query}{\text{Queries}}

\newcommand{\N}{\mathbb N}

\renewcommand{\P}{\mathcal{P}}

\newcommand{\kstar}[1]{\ensuremath{#1}-star}

\newcommand{\eps}{\varepsilon}

\newcommand{\A}{\mathcal A}

\newcommand{\beq}{\begin{equation}}
\newcommand{\eeq}{\end{equation}}

\usepackage{tikz}

\newcommand{\RACC}{\mathrm{RAC}} 

\newcommand{\remove}[1]{}
\definecolor{green1}{rgb}{0.40, 0.8, 0.2}

\newcommand\omri[1]{{\textcolor{red}{Omri: #1}}}
\newcommand\tomer[1]{{\textcolor{red}{Tomer: #1}}}

\def\({\left(}
\def\){\right)}
\makeindex

\title{
On the Instance Optimality of Detecting Collisions and Subgraphs
}

\author{
Omri Ben-Eliezer\thanks{Simons Institute for the Theory of Computing, University of California, Berkeley, USA. Part of this research was conducted while the author was at Weizmann Institute and later at Massachusetts Institute of Technology. Email: \texttt{omrib@mit.edu}.}
\and
Tomer Grossman\thanks{Department of Computer Science and Applied Mathematics,
   Weizmann Institute of Science,  Rehovot 76100, Israel. Email: \texttt{tomer.grossman@weizmann.ac.il}.}
\and
 Moni Naor\thanks{Department of Computer Science and Applied Mathematics,
   Weizmann Institute of Science,  Rehovot 76100, Israel. Email:
   \texttt{moni.naor@weizmann.ac.il}. Supported in part by grant  from the Israel
   Science Foundation (no.\ 950/16). Incumbent of the Judith Kleeman Professorial
   Chair.}
}

\date{}
\begin{document}

\maketitle

\begin{abstract}
Suppose you are given a function 
$f\colon [n] \to [n]$
via (black-box) query access to the function. You are looking to find something local, like a collision (a pair $x \neq y$ s.t.\ $f(x)=f(y)$). The question is whether knowing the `shape' of the function helps you or not (by  shape we mean that some permutation of the function is known). 
Formally, we investigate the \emph{unlabeled instance optimality} of substructure detection problems in graphs and functions. A problem is $g(n)$-instance optimal if it admits an algorithm 
$A$ satisfying that for any possible input, the (randomized) query complexity of $A$ is at most $g(n)$ times larger than the query complexity of any algorithm $A'$ which solves the same problem while holding an \emph{unlabeled copy} of the input (i.e., any $A'$ that ``knows the structure of the input''). 
Our results point to a trichotomy of unlabeled instance optimality among substructure detection problems in graphs and functions: 
\begin{itemize}
\item A few very simple properties have an $O(1)$-instance optimal algorithm.
\item Most properties of graphs and functions, with examples such as containing a fixed point or a $3$-collision in functions, or a triangle in graphs, are $n^{\Omega(1)}$-far from instance optimality. 
\item The problems of collision detection in functions and finding a claw in a graph serve as a middle ground between the two regimes. 
We show that these two properties are $\Omega(\log n)$-far from instance optimality, and conjecture that this bound is tight. We provide evidence towards this conjecture, by proving that finding a claw in a graph is $O(\log(n))$-instance optimal among all input graphs for which the query complexity of an algorithm holding an unlabeled certificate is $O\left(\sqrt{\frac{n}{\log n}}\right)$.
\end{itemize}
\end{abstract}

\section{Introduction}
\label{sec:intro}


    

Efficient detection of small structures in complex data is a fundamental challenge across computer science. 
%
In this work, we explore to what extent \emph{prior knowledge} on the input may help. 
Consider, for instance, the problem of detecting a collision in an unknown function $f\colon [n] \to [n]$ given query access to $f$. (Here, a collision in $f$ is a pair of disjoint elements $x \neq y \in [n]$ so that $f(x) = f(y)$.)
We ask the following question.

\begin{center}
\textit{How does an algorithm that knows nothing about $f$ in advance (aside from the domain size $n$) compare to an algorithm that has some prior knowledge an the structure of $f$? 
}

\end{center}
The prior knowledge we consider in this work takes the form of an \emph{unlabeled copy} of $f$ that the algorithm receives in advance as in Grossman et al.~\cite{GrossmanKN20}. That is, the algorithm receives a permutation of $f$ -- the composed function $f \circ \pi$ for some unknown permutation $\pi$ -- as an ``untrusted hint". We typically call this permutation of $f$ an \emph{unlabeled certificate}; we require the algorithm to be correct with good probability regardless of whether the hint is correct (i.e., even if $f$ is not a permutation of the unlabeled certificate). 
However, the number of queries made by the algorithm is only measured if the true input is indeed a permutation of the 
unlabeled certificate.

In the worst case, clearly $\Omega(n)$ queries are necessary, whether we know anything about the structure of $f$ or not. But are there beyond-worst-case instances where holding additional structural information on $f$ may accelerate collision detection?

\begin{definition}[instance optimality; informal] \label{def:inst_opt_inf} A randomized Las Vegas\footnote{For simplicity we consider in this paper Las Vegas randomized algorithms, but all of the results apply also to Monte Carlo type algorithms (that allow some error in the returned value).}
algorithm $A$ deciding if an unknown function $f \colon [n] \to [n]$ satisfies a property $\P$ is \emph{instance optimal} if there exists an absolute constant $\alpha$ satisfying the following. For every function $f$, and any randomized algorithm $A'$ for the same task, the following holds:
\begin{align*}
\text{Queries}_A(f) \le \alpha \cdot \max_\pi \text{Queries}_{A'}(f \circ \pi)
\end{align*}
where the $\text{Queries}_A(\cdot)$ operator refers to the expected number of queries that an algorithm $A$ makes on a certain input.

Finally, we say that $\P$ is instance optimal if there exists an instance optimal algorithm for it.

\end{definition}
Note the order of the quantifiers in the definition: for every $f$, the algorithm $A$ has to compete with an algorithm $A'$ that ``specializes'' to functions of the form $f \circ \pi$. 
In other words, an algorithm $A$ is instance optimal if it performs as well as every algorithm $A'$, that knows the structure of $f$, but not the actual labels. Note that the correctness of algorithm $A'$ is unconditional -- that is $A'$ must be correct even if the structure of $f$ doesn't match the certificate $A'$ receives.

An algorithm being unlabeled instance optimal means it always performs as well (up to a constant) as the algorithm that knows the structure of the input. If there is no instance optimal algorithm, that means there exists some function where knowing the structure of the function is helpful.
Thus, instance optimality is a strong requirement: If a property $\P$ is instance optimal that means that knowing the structure of the input function $f$ \emph{never helps}.
When a property is not instance optimal, it will sometimes be useful to discuss its ``distance'' from instance optimality.
\begin{definition}[distance from instance optimality; informal]
Consider the setting of Definition \ref{def:inst_opt_inf}.
For a function $\omega(n)$ that grows to infinity as $n \to \infty$, we say that $\P$ is $\omega$-far from instance optimality if for every algorithm $n \in \N$ and $A$ there exist a function $f$ and an algorithm $A'$ satisfying 
\begin{align*}
\text{Queries}_A(f) \ge \omega(n) \cdot \max_\pi \text{Queries}_{A'}(f \circ \pi).
\end{align*}
Similarly, $\mathcal{P}$ is $\omega$-close to instance optimality if the above inequality holds with $\leq$ instead of  $\geq$.
\end{definition}




We may now rephrase our initial question about collisions in the language of instance optimality. Is collision detection an instance optimal problem? I.e., is the property of containing a collision instance optimal? Is it far from instance optimality? Suppose that we have query access to a function $f \colon [n] \to [n]$ and are interested in finding a collision. There are two fundamental types of queries to $f$ that one can make: the first option is to query an element $x$ that we have already seen in the past, by which we mean that we have already queried some element $y$ satisfying that $f(y) = x$. This option amounts to extending a ``walk'' on the (oriented) graph of $f$. 
The second option is to query a previously unseen element $x$, which amounts to starting a new walk. The question, then, is the following: is there a universal algorithm $A$ (which initially knows nothing about $f$) for choosing when to start new walks, and which walks to extend at any given time, that is competitive with algorithms $A'$ that know the unlabeled structure of $f$?

\paragraph{Substructure detection problems.}
There are many other types of natural problems in computer science that involve small (i.e., constant-sized) substructure detection. A natural generalization of a collision is a $k$-collision (or multi-collision), where we are interested in finding $k$ different elements $x_1, \ldots, x_k$ satisfying $f(x_1) = \cdots = f(x_k)$. Fixed points, i.e., values $x$ for which $f(x) = x$, are important in local search and optimization problems, in particular for the study of local maxima or minima in an optimization setting. 

Subgraph detection in graphs is also a fundamental problem in the algorithmic literature. Motifs (small subgraphs) in networks play a central role in biology and the social sciences. In particular, detecting and counting motifs 
efficiently is a fundamental challenge in large complex networks, and a substantial part of the data mining literature is devoted to obtaining efficient algorithms for these tasks. It is thus natural to ask: is it essential to rely on specific properties of these networks in order to achieve efficiency? In other words, are subgraph detection and counting instance optimal problems?

Similarly, the problem of finding collisions is a fundamental one in cryptography. Many cryptographic primitives are built around the assumption that finding a collision for some function, $f$ is hard (e.g.\ efficiently signing large documents, commitments with little communication and of course distributed ledgers such as  blockchain). If one wants to break such a cryptographic system, should one spend resources studying the structure of $f$? If finding collisions is instance optimal, that would mean that any attempt to find collisions by studying the structure of a function is destined to be futile.

In this work we focus on the instance optimality of constant-size substructure detection problems in graphs and functions. 
Before stating our results, let us briefly discuss these data models.

%




\paragraph{Models.}
We consider two different types of data access in our work. The first type is that of functions. In this case the input is some function $f$, and the goal is to determine whether $f$ satisfies a certain property (e.g., whether it contains a collision or a fixed point). In this case the goal of an instance optimal algorithm is to perform as well as an algorithm that receives, as an untrusted hint, the unlabeled structure of the algorithm without the actual assignment of labels. Here the complexity is measured as the number of queries an algorithm makes, where each query takes an input $x$ and returns $f(x)$.

The second type of data is of graphs. Here the goal is to find a constant-sized subgraph. An instance optimal algorithm should perform as well as an algorithm that is given an isomorphism of the graph as an ``untrusted hint". For simplicity, we focus on the standard adjacency list model (e.g., \cite{Gonen2011}). Here for each vertex the algorithm knows the vertex set $V$ in advance, and can query the identity of the $i$-th neighbor of a vertex $v$ (for $v$ and $i$ of its choice, and according to some arbitrary ordering of the neighbors), or the degree of $v$. We note that all of the results also hold in other popular graph access models, including the adjacency matrix model and the neighborhood query model. 

Interestingly, graphs and functions seem closely related in our context. Specifically, the problem of finding a claw in a graph (a star with three edges) is very similar to that of finding a collision in a function, and the results we obtain for these problems are for the most part analogous.

\subsection{Main Results and Discussion}
\label{subsec:main_results}



Our main result in this paper characterizes which substructure detection problems in functions and graphs are instance optimal. 
Let us start with the setting of functions. 

A structure $H = ([h], E)$ is an oriented graph where each vertex has outdegree at most one, and we say that $f$ contains $H$ as a substructure if there exist values $x_1, \ldots, x_{h}$ such that $f(x_i) = x_j$ if and only if the edge $i\to j$ exists in $H$. (For example, a collision corresponds to the structure $([3], \{1\to 3, 2 \to 3\})$.) Finally, the property $\mathcal{P}_H$ includes all functions $f$ containing the structure $H$. Our first theorem constitutes a partial characterization for instance optimality in functions.

\begin{theorem} [Instance optimality of substructure detection in functions]\label{thm:main}
Let $H$ be a connected, constant-sized oriented graph with maximum outdegree $1$, and consider the function property $\mathcal{P}_H$ of containing $H$ as a substructure. Then $\P_H$ is
\begin{enumerate}
\item Instance optimal if $H=P_k$ is a simple oriented path of length $k$;
\item $n^{\Omega(1)}$-far from instance optimal for any $H$ that contains a fixed point, two edge-disjoint collisions, or a $3$-collision;
\item \label{funcpart3} $\Omega(\log n)$-far from instance optimal for any $H$ that contains a collision.
\end{enumerate}

\end{theorem}

Similarly, in graphs we denote by $\mathcal{P}_H$ the property of containing $H$ as a (non-induced) subgraph. Our next theorem provides a characterization for the instance optimality of subgraph detection. 

\begin{theorem} [Instance optimality of subgraph detection in graphs]\label{thm:graphconstantcharacterization}
Let $H$ be a connected, constant-sized graph with at least one edge. Then $\mathcal{P}_H$ is:
\begin{enumerate}
\item Instance optimal if $H$ is an edge or a wedge (path of length 2);
\item $n^{\Omega(1)}$-far from instance optimal if $H$ is any graph other than an edge, a wedge, or a claw (a star with $3$ edges);
\item \label{graph:part3} $\Omega(\log n)$-far from instance optimal when $H$ is a claw.
  \end{enumerate}
\end{theorem}


\paragraph{Almost instance optimality of claws and collisions?}
While we provide a full characterization of those substructures (or subgraphs) $H$ for which $\P_H$ is instance optimal, there remains a notable open problem: is the problem of containing a collision (in functions) or a claw (in graphs) ``almost instance optimal'', e.g., is it $O(\log n)$-close to instance optimality?

The problems of finding a collision in a function and detecting a claw in a graph are closely related and seem to be similar in nature (see Section \ref{sec:3}).
We conjecture that both of these problems are close to being instance optimal. 

\begin{conjecture} \label{conjcollision}
There exists an algorithm $A$ for collision detection (in functions $f \colon [n] \to [n]$) that is $O(\log n)$-close to instance optimality.
\end{conjecture}
\begin{conjecture}
\label{conjgraph}
Determining if a graph contains a claw is $O(\log n)$-close to instance optimality.
\end{conjecture}

While we are not yet able to prove the conjectures in full generality, we provide initial evidence toward the correctness, at least in the graph case. Specifically, we prove Conjecture~\ref{conjgraph} for graphs in which claw detection is ``easy'' with a certificate, that is, can be done in up to $O\left(\sqrt{\frac{n}{\log n}}\right)$ queries.


\begin{theorem}[informal; see Theorem \ref{thm:near_instance_optimality}]
\label{thm:inst_opt_graphs_intro}
The graph property of containing a claw is $O(\log n)$-instance optimal when restricted to inputs that require $o\left(\sqrt{\frac{n}{\log n}}\right)$ queries in expectation for an algorithm with an unlabeled certificate.
\end{theorem}
While the result was phrased for undirected graphs, it carries on also for collision detection in functions $f \colon [n] \to [n]$, in the case where the algorithm is allowed to go both ``forward'' (i.e., for an $x$ to retrieve $f(x)$) and ``backward'' (i.e., for $x$ to retrieve elements of the inverse set $f^{-1}(x)$). See the paragraph below on model robustness for further discussion.


We conjecture that the same algorithm we use to show near instance optimality in the regime of Theorem~\ref{thm:inst_opt_graphs_intro}
is also near instance optimal in the general regime. The algorithm $A_{\text{all-scales}}$ is roughly defined as follows. $A_{\text{all-scales}}$ maintains $m = O(\log n)$ parallel ``walks'' $W_1, \ldots, W_m$ at different ``scales'', where in each round (consisting of a total of $m$ queries) $A_{\text{all-scales}}$ adds one step to each of the walks. We try to extend each $W_i$ until it reaches length $2^i$ or until it has to end (either because of finding a collision/claw or due to reaching the end of a path or closing a cycle). In the case that $W_i$ reaches length $2^i$, we ``forget'' it and restart $W_i$ at a fresh random starting point. 

\paragraph{The challenge of merging walks.}
The only barrier to proving the above two conjectures in the general case seems to be our current inability to deal with ``merging'' walks in the algorithm. Any algorithm for collision detection in functions, or claw detection in graphs, can be viewed as maintaining a set of walks. In each step we either choose to start a new walk by querying a previously unseen vertex, or extend an existing walk by querying its endpoint (or one of its two endpoints, in the graph case). The event of merging corresponds to the case that two disjoint walks $W$ and $W'$ meet, resulting in the creation of a longer walk $W \cup W'$. Our proof of Theorem~\ref{thm:inst_opt_graphs_intro} shows the instance optimality of claw detection in the regime where merging is unlikely to happen during the algorithm's run.

\label{subsec:discussion}

\paragraph{Model robustness.}
Throughout the paper we chose to focus on specific models for convenience. However, all our results are model robust and apply in many ``natural" models. In particular, in the case of functions we chose to work on the model where an algorithm can only go forward. That is, an algorithm can query $f(x)$ in a black box manner, and doesn't have the capability to make inverse/backward ($f^{-1}(x)$) queries. Similar characterization results to the graph case also apply if an algorithm can walk backwards; in fact, the model where walking backward is allowed seems to serve as a middle ground between our models for graphs and functions, in the sense that we deal with directed graph properties but are allowed to move in the graph as if it were undirected.

For convenience we wrote all our results for Las Vegas randomized algorithms. All the results in this paper also apply if we require the algorithm to be a Monte Carlo randomized algorithm, i.e., one that is allowed to err with constant probability.

In graphs, we use the popular adjacency list model (which allows sampling random vertices, querying a single neighbor, or querying the degree of a vertex) for data access. The same characterization results also apply under other types of data access, such as the adjacency matrix model or the neighborhood query model (where querying a node retrieves all of its neighbors at once).




\subsection{Technical Overview: Collisions and Fixed Points}

In this section we give an overview of our main ideas and techniques. Many of the ideas are shared between functions and graphs; we chose to present the main ideas for a few canonical problems, such as fixed point and collision detection in functions, and claw detection in graphs. 

Showing the polynomial separation for most graph and function properties amounts, roughly speaking, to providing constructions where a certain substructure is hidden, but where certain hints are planted along the way to help the algorithm holding a certificate to navigate within the graph. Given the constructions, which are themselves interesting and non-trivial, it is not hard to prove the separation. As an example of a polynomial separation construction and result, we discuss the case of a fixed point in functions. For more general statements and proofs regarding these separations, please refer to Sections \ref{sec:4} (for functions) and \ref{sec:5} (for graphs).

The $\Omega(\log n)$-separation for claws and collisions is the most technically involved ``lower bound'' type contribution of this paper. Unlike the polynomial separation results, where the core idea revolves around describing the ``right'' way to hide information, here the construction is more complicated (roughly speaking): the trick of planting hints that allow the algorithm to navigate does not work well, and our arguments rely on the observation that it is sometimes essential for an algorithm without a certificate to keep track of multiple different scales in which relevant phenomena may emerge, compared to an algorithm with a certificate that knows in advance which of the scales is relevant.
The proof is also more challenging, requiring us to closely track counts of intermediate substructures of interest. For the sake of the current discussion, we focus on collision detection, but the proof (and construction) for claws is very similar; see Section \ref{sec:3}.

Before diving into the ideas behind fixed point and collision detection, let us briefly mention the simplest component in the characterization: an instance optimal algorithm for finding a path of length $k$. The algorithm chooses a random value and evaluates the function $k$ times on successive values to see if a path of length $k$ emerges (and not a smaller cycle, or a smaller path ending in a fixed point). This is repeated until a path is found or all values have been exhausted.  It is instance optimal, since knowing the structure of the function does not help; stopping after less than $k$ steps is meaningless, since it only saves us a constant fraction of the queries.

\subsubsection{Fixed point detection: Polynomially far from instance optimality}

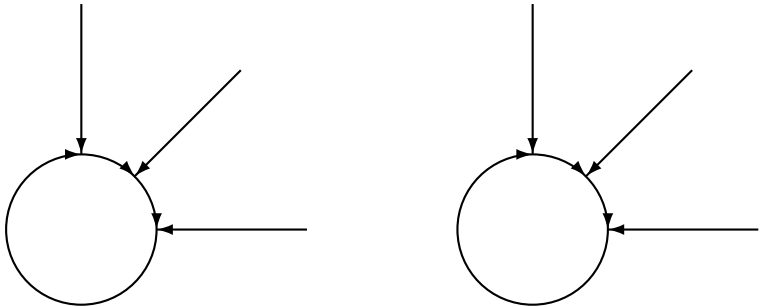
\begin{figure} 
    \centering
\setlength{\unitlength}{1cm}
\thicklines

\begin{picture}(10,6)
\put(3,2){\vector(-1,0){0}}
\put(3,2){\line(1,0){2}}
\put(2,3){\line(0,1){2}}
\put(2,3){\vector(0,-1){0}}
\put(2,3){\vector(1,0){0}}

\put(2.71,2.71){\line(0.71,0.71){1.41}}

\put(2.71,2.71){\vector(-0.71,-0.71){0}}
\put(2.71,2.71){\vector(0.71,-0.71){0}}

\put(2,2){\circle{2}} 
\put(3,2){\vector(0,-1){0}}

\put(9,2){\line(1,0){2}}
\put(9,2){\vector(0,-1){0}}
\put(9,2){\vector(-1,0){0}}

\put(8,3){\line(0,1){2}}
\put(8,3){\vector(0,-1){0}}
\put(8,3){\vector(1,0){0}}

\put(8.71,2.71){\line(.71,.71){1.41}}
\put(8.71,2.71){\vector(-0.71,-0.71){0}}
\put(8.71,2.71){\vector(0.71,-0.71){0}}

\put(8,2){\circle{2}}
\end{picture}

\caption{There are $n^{1/4}/\log n$ cycles, where each cycle is of length $n^{3/4}$. Each path entering a cycle is of size $n^{1/4}$. The distance between every two paths on the $i$-th cycle is $p_i$.} \label{fig:fixed}
\end{figure}
We give an overview of the proof that finding a fixed point is polynomially far from instance optimality. Small variations of the constructions can be used to show that the same is true for any structure containing a fixed point, a $3$-collision, or two edge-disjoint collision.

In order to obtain such a result 
we provide a distribution of functions that have several fixed 
points  with a secret parameter so that an algorithm with a certificate (knowing the  parameter in this case) can find a fixed point in $n^{\frac{3}{4}}$ queries while any algorithm that does not know the secret parameter (i.e.\ without  a certificate) requires $\tilde\Omega(n)$ queries to find a fixed point.  

The idea is to construct a function $f$ with $\tilde{\theta}(n^{1/4})$ cycles of size roughly $n^{3/4}$, where one random value $x$ in one of the cycles is turned into a fixed point (which effectively turns the said cycle into a path ending at $x$). It is quite clear that for such a distribution finding the fixed point take time $\tilde\Omega(n)$. But we want to add some information or hint that will allow a certificate holder to find out which is the ``correct" cycle.

To give such a hint we add to each cycle  many paths of length $n^{1/4}$ entering it. The distance between two paths entering the $i$th cycle is some (unique)  prime $p_i$ where $p_i$ is of size roughly $n^{1/4}$ (so roughly $n^{1/2}$ paths enter the cycle). See~\Cref{fig:fixed} for a drawing of this construction.

The hint is the value $p_i$ associated with the unique cycle that ends up with a fixed point. The algorithm (with the hint) we propose will check many (about $\sqrt n$) `short' (length $n^{1/4}$) paths and see when they collide with another set of paths that is supposed to be on the cycles (these are $n^{1/4}$ `long' paths of  length $\sqrt n$).  Once our algorithm finds three paths entering the same cycle which are of distances that are all a multiple of $p_i$, the algorithm will conclude that this is the unique path that at its end the fixed point resides and will continue on the path. On the other hand,  for any algorithm that does not know which of the $p_j$'s is the chosen one and hence the  which path ends in a fixed point, each $x$ residing in a cycle is equally likely to be a fixed point, and thus the algorithm requires $\tilde{\Omega}(n)$ queries in expectation. 


\subsubsection{Finding Collisions:  $\Omega(\log n)$ far from instance optimality}

The distribution constructed above will not work for collision detection, since functions generated according to this distribution will inherently have many collisions. Below we describe a substantially different construction demonstrating that collision detection is (at least) logarithmically far from instance optimality.
We note that the same proof outline, and same construction idea can also be used to show that finding a claw in a graph is not instance optimal.

In order to obtain such a result 
we provide a distribution of functions that have several collisions, again,  with a secret parameter, so that an algorithm with a certificate (knowing the  parameter in this case) can find a collision in $n^{c}$ queries for some constant $c<1/2$, while any algorithm that does not know the secret parameter (i.e.\ without  a certificate) requires $\Omega(n^{c} \log n)$ queries to find a collision.

The hard distribution is as follows: there are $\log n$ length scales. For scale $i$ we have $n/{2.2^i}$ cycles, each  of length $2^i$ (note that the total number of nodes in all cycles is $O(n)$). For a uniformly randomly chosen scale $t$ we turn $n^{1-c}/{1.1^t}$ of the cycles to be a path ending in a loop of size $2$ at the end (this is a collision). 

The secret parameter is the value of $t$. The algorithm with a certificate simply picks a value at random and follows the path it defines for $2^t$ steps. The algorithm stops if (i)  a collision is discovered or (ii) the path has reached length $2^t$ without a collision  or (iii) the path created a cycle of length $2^i < 2^t$. 
The probability of picking a node on a good path (one ending in  a collision of length $2^t$)  is 
$$\frac{n^{1-c} \cdot 2^t}{n \cdot 1.1^t} $$
(since there are $n^{1-c}/{1.1^t}$ such cycles, each of size $2^t$).
The cost (in terms of queries)  of picking a value on the wrong size cycle, say of size $2^i$, is  $\min(2^i,2^t)$. It is possible to show that the total expected work when picking the wrong value is $O(2^t/1.1^t)$.\footnote{the constant $1.1$ is a bit arbitrary, and other constants larger than 1 will also work.} Therefore the expected amount of work until a good path is found is    
$$\frac{2^t}{1.1^t} \cdot \frac{1.1^t \cdot n }{n^{1-c} \cdot 2^t} = n^{c}.$$ 
The result is $O(n^{c})$ queries in expectation. 

We next show that any algorithm that does not know $t$ requires $\Omega(n^{c} \log n)$ queries, which results in a logarithmic separation from the algorithm with a certificate. In essence, this means that the algorithm needs to spend a substantial effort at all possible scales (instead of just one scale, $t$) in order to find the collision. 

Consider an algorithm without a certificate, and suppose that we choose the secret parameter $t$ in an online manner as follows. Our initial construction starts with $n/{2.2^i}$ cycles of length $i$. For each such $i$, we pick $n^{1-c}/1.1^i$ of the cycles of length $2^i$, and color one of the nodes in each such cycle by red (call these points the ``$i$-red points''. Note that at this time we have no information whatsoever on $t$. Now, each time that a red point on some cycle of length $2^i$ is encountered, we flip a coin with an appropriate probability (which initially is of order $1/\log n$) to decide whether the current value of $i$ is the secret parameter $t$ or not. If it is, then we turn all $i$-red points (for this specific value of $i$) into collisions as described above, and remove the color from all other red points (in paths of all possible lengths). Otherwise, we remove the color from all $i$-red points (for this specific $i$) and continue.

It turns out that this construction produces the same distribution as we described before (where $t$ was chosen in advance). However, it can also be shown that to find a collision with constant probability, $\Omega(\log n)$ red points need to be encountered along the way. The rest of the analysis provides an amortized argument showing that the expected time to find each red vertex by any algorithm is $\Omega(n^{c})$. The main idea of the amortized analysis (which we will not describe in depth here) is to treat cycles in which we made many queries -- at least a constant fraction of the cycle -- differently from cycles where we made few queries. For cycles of the first type, the probability to find a red point (if one exists) is of order $2^i / n^c$, but the amount of queries that we need to spend is proportional to $2^i$. For cycles of the second type, each additional query only has probability $O(1 / n^c)$ to succeed finding a red point, but the query cost is only $1$. In both cases, the rate between the success probability and the query cost is of order $1/n^{c}$.

\subsubsection{Finding claws: $O(\log n)$-close to instance optimality in merging-free regime}
The proof that collision detection is $\Omega(\log n)$-far from instance optimality extends very similarly to claw detection in graphs. We next show that this bound is tight in the ``low query'' regime where $G$ admits an algorithm for claw detection (with a certificate) using $q \leq \alpha \sqrt{\frac{n}{\log n}}$ queries, for a small constant $\alpha > 0$.

Every claw-free graph is a union of disjoint cycles and paths. Thus, every algorithm for finding a claw can be viewed as maintaining a set of walks of several types: Some of the walks may have closed a cycle; others may have reached one or two ends of the path $P$ that the walk resides in. All other walks simply consist of consecutive vertices in the interior of some path in $G$. 

Clearly, walks of the first type -- those that have closed a simple cycle -- cannot be extended to find a claw. Our first part of the proof is a reduction eliminating the need to consider walks of the second type (i.e., ones that have reached at least one endpoint). Specifically, we show that for every graph $G$ on $n$ vertices there is a graph $G'$ on $n+2$ vertices satisfying several properties:
\begin{itemize}
\item In all simple paths in $G'$, either both endpoints have degree $1$ or both are degree $3$ or more (i.e., are claw centers).
\item The query complexity of detecting a claw in $G'$ is equal, up to a multiplicative constant, to that of $G$. This is true both with or without an unlabeled certificate.
\item The construction of $G'$ from $G$ is deterministic. Thus, an unlabeled certificate for $G'$ can be constructed if one has an unlabeled certificate for $G$.
\end{itemize}
The construction is very simple: we add two new vertices and connect them to all degree-$1$ vertices (and to each other, if needed). 

\paragraph{Merging without claws requires $\Omega\left(\sqrt{\frac{n}{\log n}}\right)$ queries.}
The second part of our argument shows that one cannot make two walks merge in the first $\alpha \sqrt{{n}/{\log n}}$ queries (with constant probability and for some small constant $\alpha > 0$) without finding a claw beforehand. The proof relies on an advanced birthday paradox type analysis that is suitable for adaptive algorithms. For the purpose of this part, one may consider $G$ as a union of paths and cycles (without any claws). We bucket these paths and cycles into $O(\log n)$ groups, where in each group all paths and cycles are of roughly the same length -- within a multiplicative factor of $1.1$ from each other. Suppose that after each ``new'' (previously unseen) vertex is queried, the algorithm immediately knows to which bucket this vertex (and the corresponding walk emanating from it) belongs. We show that even in this case the lower bound holds.

Focusing on a specific bucket, let $\mathcal{W}$ be the set of all walks in this bucket at some point in the algorithm's run. An assignment of walks to ``locations'' within the bucket is considered valid if no two walks intersect.
We argue that the set of locations of walks is uniformly random among all sets of valid configurations. 
Similarly to our analysis of the $\Omega(\log n)$ separation (see the previous subsection), our next step in this part deals separately with ``short walks'' and ``long walks''. Very roughly speaking, our proof shows that walks have sufficient ``degree of freedom'' so that their probability to merge will be very small, even if provided that they lie in the same bucket.
We omit the precise details of the analysis from this overview, and point the reader to Section \ref{sec:nomerge}.

\paragraph{Asymptotic stochastic dominance of $A_{\text{all-scales}}$}
The third and last part of the argument shows that the algorithm $A_{\text{all-scales}}$ mentioned in Section~\ref{subsec:main_results} stochastically dominates any other algorithm (asymptotically) in the following sense. Conditioned on an algorithm $A$ (making $q = O\left(\sqrt{\frac{n}{\log n}}\right)$ queries) not encountering any merging walks during its operation, the algorithm $A_{\text{all-scales}}$ is at least as likely to find a claw, while using a slightly larger amount of roughly $4q \log n$ queries. 

Recall first how $A_{\text{all-scales}}$ is defined. For $0 \leq i < \log n$ let $A^{(i)}$ be the algorithm which repeatedly does the following: pick a random vertex in $G$; make a bidirectional walk from it for $2^{i+1}$ steps, or until a claw is found (leading to a ``win'') or an endpoint is found (leading to an early termination). 
$A_{\text{all-scales}}$ maintains one copy of each $A^{(i)}$ (for a total of $\log n$ copies), and alternates between them: each round of $A_{\text{all-scales}}$ makes $\log n$ queries, one for each $A^{(i)}$.

Now let $A$ be any algorithm operating on graphs on $n$ vertices.
Consider the event $E_i = E_i(G, q, A)$ that $A$ finds a claw (for the first time) after at most $q$ queries through the following process. $A$ queries a ``new'' vertex $v$ of distance between $2^{i}$ and $2^{i+1}-1$ from the claw center $w$, and finds it by completing a walk from $v$ to $w$. 
We claim that the event that $A$ finds a claw is equal to the union of the events $E_i$ (for $0 \leq i \leq \log n$).

The proof goes through a careful coupling argument between $A$ and any fixed $A^{(i)}$ (separately). Through the coupling, we may assume that $A$ and $A^{(i)}$ have access to the same source of randomness generating ``new'' vertex queries, that is, the $j$-th new vertex starting a walk $W_j$ in $A$ is also the $j$-th new vertex in $A^{(j)}$. We may further assume, by symmetry considerations, that $A$ respects an ``older first'' principle: if two walks $W$ and $W'$ have exactly the same ``shape'' (within $G$) at some point, then $A$ will prefer to extend the walk among them that is older. Now, suppose that $A$ finds the claw in its walk $W_j$, of distance between $2^i$ and $2^{i+1}-1$ from $v_j$. By the ``older first'' principle, this implies that for all vertices $v_{j'}$ with $j' < j$ that are at least $2^i$ away from an endpoint (call these values of $j'$ good), $A$ must have walked for at least $2^i$ from $v_{j'}$ so far. For all other values of $j' < j$ (call these bad), $A$ walked a number of steps that is at least the distance from an endpoint.
In contrast, $A^{(i)}$ makes up to $4 \cdot 2^{i}$ queries around each good vertex, and up two times as many queries as $A$ around each bad one. 

\subsection{Related Work}
The term instance optimality was coined by Fagin, Lotem and Naor \cite{FaginLN03}. If an algorithm always outperforms every other algorithm (up to a constant), particularly the algorithm that is aware of the input distribution, it is defined as being instance optimal. This definition is very strict, and thus there are numerous situations in which it cannot be meaningfully attained. As a result, several works (including ours) address this by necessitating that an instance optimal algorithm be competitive against a certain class of algorithms (for example, those that know an a permutation of the input, rather than the input itself). This notion of instance optimality is sometimes referred to as ``unlabeled instance optimality".

\paragraph{Unlabeled instance optimality.} Afshani, Barbay, and Chan~\cite{AfshaniBC17} considered and presented unlabeled instance-optimal algorithms for locating the convex hull and set maxima of a set of points. Valiant and Valiant \cite{ValiantV16} developed an unlabeled instance optimal algorithm for approximating a distribution using independent samples from it (with the cost function being the number of samples). Later \cite{ValiantV17}, they provided an algorithm for the identity testing problem. Here the problem is  determining, given an explicit description of a distribution, whether a collection of samples was selected from that distribution or from one that is promised to be far away. More recent works on instance optimality in distribution testing include, for example, the work of Hao et al.~\cite{HaoOSW18,HaoOrlitsky2020}.

Grossman, Komargodski, and Naor examined unlabeled instance optimality in the query model \cite{GrossmanKN20}. Their work relaxes the definition of instance optimality by no longer requiring an optimal algorithm to compete against an algorithm that knows the entire input, but rather against an algorithm that knows \emph{something} about the input. Arnon and Grossman \cite{ArnonG21} define the notion of min-entropic optimality, where instead of relaxing the ''for-all" quantifier over algorithms, they relax “for-all” quantifier over inputs. That is, for an
algorithm to be optimal it is still required to perform as well any other algorithm; however it is no longer required to be optimal on every input distribution, but rather only on a certain class of inputs. 

\paragraph{Instance optimality in graphs.}
Subgraph detection and counting has not been thoroughly investigated from the perspective of instance optimality; establishing a unified theory of instance optimality in this context remains an intriguing open problem. However, instance optimality has been investigated for other graph problems of interest. For example, Haeupler, Wajc and Zuzic \cite{HWZ21} investigate instance optimality and a related notion called universal optimality in a family of classical and more ``global'' distributed graph problems such as computing minimum spanning trees and approximate shortest paths. 

\paragraph{Strong instance optimality.} The original, robust definition of instance optimization calls for an algorithm to be superior to every other algorithm.  For getting the top $k$ aggregate score in a database with the guarantee that each column is sorted, \cite{FaginLN03} provided an instance-optimal algorithm. Demaine, L\'{o}pez-Ortiz, and Munro \cite{DemaineLM00} provided instance-optimal algorithms for locating intersections, unions, or differences of a group of sorted sets. Baran and Demaine \cite{BaranD04} showed an instance optimal algorithm for finding the closest and farthest points on a curve. Grossman, Komargodski and Naor \cite{GrossmanKN20} and Hsiang and Liu \cite{HsiangL23} studied instance optimality in the decision tree model. 



 \paragraph{Cryptography and complexity.}
 The problems of finding a constant sized structure in $f$, where $f$ is a total function guaranteed to contain the structure at hand has been studied extensively and is a fundamental problem in computational complexity and there are complexity classes in TFNP defined around it \cite{MegiddoP91,Papadimitriou94}. We note that we can slightly change the functions in our paper to also make them total problems, and all our proofs will still hold. 
 
As mentioned above, the problem of finding collisions is a fundamental one in cryptography. The standard definition is that of a collision resistant hash (CRH),  where finding a collision is a computationally hard problem. Such functions are very significant for obtaining efficient signature schemes and commitment to large piece of data using little communication. But other related structures are also considered in the literature: for instance, functions where it is hard to find multiple collisions~\cite{KomargodskiNY18}. 
 

\subsection{Organization}
In Section \ref{sec:prelims} we formally define the computational models we use as well as the notions of an unlabeled certificate (i.e., the ``untrusted hint'') and instance optimality. In Section \ref{sec:3} we prove that finding a collision in functions, and a claw in a graph is not instance optimal. In Section \ref{sec:4} we prove that functions that are not subsets of the collisions of two paths, followed by an additional path is polynomial far from being instance optimal. Such properties include many fundamental structures such as finding a fixed point, or a multi-collision. This gives us a complete characterization in the function model. In~\Cref{sec:5} we complete the full characterization for the model defined on graphs. We do this by proving that finding a path of length one or two is instance optimal, and that determining if a graph contains a subgraph $H$ is polynomially far from being instance optimal, unless $H$ is a path of length 1 or 2, or a claw. Finally, in~\Cref{sec:AlmostIO} we prove that finding a claw in a graph is $O(\log(n))$ close to being instance optimal among graphs for which finding a claw or collision can be done in $O(\sqrt{n} /\log(n))$, by an algorithm with an unlabeled certificate.

\section{Preliminaries}
\label{sec:prelims}

\subsection{Functions}\label{sec:prelimsfunc}


Let $n \in \N$.
A \emph{property} $\P$ of functions is a collection of functions $f \colon [n] \to [n]$ that is closed under relabeling. That is, if $f \in \P$ then $f \circ \pi \in \P$ for any permutation $\pi \colon [n] \to [n]$. 
We sometimes say that $f$ \emph{satisfies} $\P$ when $f \in \P$.

In this work we will be interested in the property of containing some constant size substructure $H$ (e.g., a collision or a fixed point). Let $H$ be an oriented graph with $h$ vertices. Suppose further that the outdegree of each vertex in $H$ is at most $1$.\footnote{Note that a function can be viewed as an oriented graph where the outdegree is always equal to one, hence $H$ can appear as a substructure in such a function if and only if the outdegrees are at most $1$.} The property $\mathcal{P}_H$ consists of all functions $f \colon [n] \to [n]$ satisfying the following. There exist $h$ disjoint elements $x_1, \ldots x_h \in [n]$ and a mapping between $V(H)$ and $\{x_1, \ldots, x_h\}$, so that $H$ contains an edge between $u$ and $v$ if and only if $f(x_u) = x_v$, where $x_u$, $x_v$ are the mappings of $u$ and $v$, respectively.

%
A Las Vegas (randomized) algorithm for the property $\mathcal{P}$ in the query model is a randomized decision tree that determines membership in the property $\P$ with probability $1$ (i.e., it is always correct, and the quantity of interest is the number of queries the algorithm requires).
Given a Las Vegas randomized algorithm $A$ (which knows $n$) with random seed $r$ and given $f \colon [n] \to [n]$,
denote by $\Query^\P_A(f,r)$ the amount of queries that $A$ makes when evaluating if $f$ satisfies a property using the random seed $r$. Usually when the property $\mathcal{P}$ is clear from context we omit it from the notation.

We write $\Query^\P_A(f) = \E_{r} \Query^\P_A(f,r)$ (or $\Query_A(f)$) to denote the expected number of queries that the algorithm $A$ makes over input $f$, where the expectation is taken over all possible random seeds.

\begin{definition}[Unlabeled Certificate Complexity]
  The \emph{Unlabeled Certificate} complexity of a property $\P$, and function $f$ is:
\begin{align*}
    \RACC(\P, f) =  \min_{A \in \A_\P} \max_{\pi} 
    \Query_A(f \circ \pi),
\end{align*}
where $\A_{\P}$ is the set of all Las Vegas algorithms for evaluating if $f$ satisfies property $\P$, and $\pi$ ranges over all permutations of $[n]$.
\end{definition}

So far, we have considered only properties of functions of a given size $n$. Our definition of instance optimality is asymptotic in its nature and so we extend the definition of a property by allowing it to have functions of different sizes. Suppose that  $\mathcal{P}$ is a property which contains graphs of all sizes $n \geq N$, for some constant $N$. 
We can then define a corresponding sequence of algorithms $\{A_n\}_{n \geq N}$, where $A_n$ is responsible for graphs of size $n$.
\begin{definition}[instance optimality]
  A sequence of properties $\P=\{\P_n\}_{n\in \N}$ invariant under a relabeling is \emph{instance optimal} if there exist an absolute constants $c > 0$,
  and a sequence
  $\A = \{A_n\}_{n\in \N}$, where each $A_n$ is a Las Vegas algorithm for $\mathcal{P}_n$, such that on every input $f \colon [n] \to [n]$, it holds that
  \begin{align*}
    \Query^{\P_n}_{A_n}(f) \leq c \cdot \RACC(\mathcal{P}_n,f)
  \end{align*}
\end{definition}
Next we present the analogous definition for being far from instance optimality.
\begin{definition}[\emph{$\omega$-far} from instance optimality]
Let $\omega\colon \N \to \N$ denote a function that grows to infinity as $n \to \infty$.
We say that a sequence of algorithms $\{A_n\}_{n \in \N}$ evaluating if a sequence of functions $\{f_n\}_{n \in \N}$ (where $f_n \colon [n] \to [n]$) satisfies a sequence of properties $\P = \{\P_n\}_{n \in \N}$ is \emph{$\omega$-far from instance optimal} if there exists a constant $N$ where for all $n \ge \N$ it holds that:
  \begin{align*}
   \Query^{\P_n}_{A_n}(f_n) \ge \omega(n) \cdot \RACC(\P_n,f_n).
  \end{align*}
  We say that the sequence of properties $\{\P_n\}$ is \emph{$\omega$-far from instance optimal} if any sequence of algorithms $\{\A_n\}_{n \in \N}$ evaluating it is $\omega$-far from instance optimal.
\end{definition}
In particular, a property $\P$ (or more precisely a sequence $\{\P_n\}$ of properties of functions $f \colon [n] \to [n]$, for any $n$) is polynomially far from instance optimal, if it is $\omega$-far for some $\omega(n) = n^{\Omega(1)}$ polynomial in $n$.

\subsection{Graphs}

A graph property $\P$ is a collection of graphs that is closed under isomorphism. That is, if $G = (V,E) \in \P$ and $\pi \colon V \to V$ is a permutation, then the graph $G^\pi = (V,E^\pi)$ where $(u,v) \in E$ if and only if $(\pi(u), \pi(v)) \in E^{\pi}$ satisfies $G^{\pi} \in \P$.

Here we consider the adjacency list query model. We assume that the vertex set $V$ is given to us in advance. Given a single query, an algorithm can either (i) find the degree $d_v$ of $v$, or (ii) find the $i$-th neighbor of $v$ (in some arbitrary ordering). We note that other variants of the adjacency list model in the literature also allow pair queries, that is, given $u,v\in V$ the algorithm can ask whether there is an edge between $u$ and $v$. Our results hold word for word also in this variant.




Definitions of instance optimality are analogous to~\Cref{sec:prelimsfunc}, except here the unlabeled certificate is an isomorphism of the graph.

\begin{definition}[Unlabeled certificate complexity]
  The \textbf{randomized unlabeled certificate} complexity of a graph property $\P$ with respect to a graph $G$ is defined as follows.
  \begin{align*}
    \RACC(\P, G) =  \min_{A \in \A_\P} \max_{\pi \in \Gamma} 
    \Query^\P_A(\pi(G)),
  \end{align*}
  where $\Gamma$ is the set of all permutations of the vertex set, and $\A_\P$ is the set of all Las Vegas randomized algorithms that always evaluate membership in $\P$ correctly.
\end{definition}

\begin{definition}[instance optimality]
  A sequence of graph properties $\P=\{\P_n\}_{n\in \N}$ is \emph{instance optimal} if there exist a constant $c > 0$ and a sequence of Las Vegas randomized algorithms
  $\A = \{A_n\}_{n\in \N}$ for $\P$, such that on every input $G$ on $n$ vertices, it holds that
  \begin{align*}
    \Query^{\P_n}_{A_n}(G) \leq c \cdot \RACC(\mathcal{P}_n,G)
  \end{align*}
\end{definition}

\begin{definition}[$\omega$-far from instance optimality]
Let $\omega\colon \N \to \N$ denote a function that grows to infinity as $n \to \infty$.
A sequence of algorithms $\{A_n\}_{n \in \N}$ evaluating if a sequence of graphs $\{G_n\}_{n \in \N}$ (where $G_n$ is a graph of order $n$) satisfies a sequence of properties $\P = \{\P_n\}_{n \in \N}$ is \emph{$\omega$-far from instance optimal} if there exists a constant $N$ where for all $n \ge \N$ it holds that:
  \begin{align*}
   \Query^{\P_n}_{A_n}(G_n) \ge \omega(n) \cdot \RACC(\P_n,G_n).
  \end{align*}
  We say that the sequence of properties $\{\P_n\}$ is \emph{$\omega$-far from instance optimal} if any sequence of algorithms $\{\A_n\}_{n \in \N}$ evaluating it is $\omega$-far from instance optimal.
\end{definition}

We conclude with standard graph theory terminology. A simple path with $k \geq 2$ vertices (in an undirected graph) is a collection of disjoint vertices $v_1, \ldots, v_k$ where $v_i$ is connected to $v_{i+1}$ by an edge for each $i=1,2,\ldots,k-1$. A simple cycle is defined similarly, but with $v_k$ also connected to $v_1$. 
Finally, the claw graph plays a central role in this work.
\begin{definition}[Claw]
The \emph{Claw graph}, $S_3$, is a \kstar{3}. That is, a four vertex graph consisting of a single vertex, of degree three, which is connected to three vertices each with degree one.
\end{definition}

\section{Collisions and Claws: Logarithmic Separation} \label{sec:3}


In this section we formally present and analyze our construction proving~\Cref{thm:main}~\Cref{funcpart3} and~\Cref{thm:graphconstantcharacterization}~\Cref{graph:part3}: that detecting collisions (in functions $f\colon [n] \to [n]$) and claws (in graphs) is not instance optimal. 

\begin{theorem}[~\Cref{thm:graphconstantcharacterization}~\Cref{graph:part3} Reworded]
The property $\P_{S_3}$ of containing a claw is $\Omega(\log(n))$-far from instance optimality.

\end{theorem}

\begin{theorem}[~\Cref{thm:main}~\Cref{funcpart3} Reworded]\label{thm:3.2}
Fix $a,b,c \in \N$. Let $H = H_{a,b,c}$ denote the oriented graph containing two paths of length $a$ and $b$ which collide in a vertex, followed by a path of length $c$. 
The function property $\P_H$ is $\Omega(\log(n))$-far from instance optimal. 
\end{theorem}

These two cases (i.e., claws in graphs and collisions in functions) are very similar and the proof that they are not instance optimal is almost identical. Thus, for the majority of the section we focus on the case of claws in graphs. At the end of the section we describe the minor adaptations required for the case of collisions in functions.



We start by presenting the construction for claws. In~\Cref{sec:clawstocollision} we adapt the construction for collisions in functions. Here and in the rest of the paper, we do not try to optimize the constant terms. In particular, the constant $c=1/10$ appearing in the exponent of the query complexity is somewhat arbitrary; the same construction essentially works for any $c < 1/2$ (and with some adaptations it can be made to work for larger values of the constant $c$).

\begin{construction}
\label{const:first_const}
Consider the following process for generating a graph over the vertex set $[n]$, which starts with an empty graph and gradually adds edges to it.

\begin{itemize}
    \item For each integer $\frac{1}{1000}  \log n \leq i \leq \frac{1}{100} \log n$, pick $a_i = n / 2.2^i$ uniformly random disjoint simple paths of length $2^i$ in the graph. 
    
    \item Pick a uniformly random integer $\frac{1}{1000} \log n \leq t \leq \frac{1}{100} \log n$, which we consider as the ``good'' index. Pick a random collection $\P_t$ of $b_t = n^{9/10} / 1.1^t$ of the paths of length $2^t$. Apply to each path $P \in \P_t$ the following transformation:
    let $u_P$ and $v_P$ denote the two ends of the path. Now connect $u_P$ to two isolated vertices, and $v_{P}$ to two other isolated vertices. This turns $P$ into a tree of size $2^t + 4$ built from a long path and two claws, one at each end of the path.
    
    \item All vertices that do not participate in any of the above structures remain isolated. 
\end{itemize}

\end{construction}
We claim that an algorithm holding a certificate requires only $O(n^{1/10})$ queries to find a claw. Since $t$ is known from the certificate, the strategy is simply to only try walks of length $2^t$.

\begin{lemma}
$\E_{G \leftarrow \Delta} \RACC(\P, G) = O(n^{1/10})$.
\end{lemma}

\begin{proof}
We show the following stronger claim: For any $G \in \Delta$, $\RACC(\P, G) \in O(n^{1/10})$. 

The algorithm repeats the following until a claw is found. It picks a point at random, and walks at an arbitrary direction for $2^t$ steps or unless the walk cannot continue anymore, i.e., due to reaching the end of a path. The correctness of the algorithm is immediate. We show that if the true function matches the certificate then the algorithm makes $O(n^{1/10})$ queries in expectation.

The probability that a random point will fall on a path of length $2^i$ is $\frac{a_i \cdot 2^i}{n} = \frac{1}{1.1^i}$. The number of queries made, if we land on a path of length $2^i$ is bounded by $\min(2^i, 2^t)$. 

Thus the expected number of queries made by our algorithm every time it picks a point and walks until the path ends or until the walk reaches a length of $2^t$ is at most
\beq \sum_{i=0}^{t} 2^i \cdot \frac{1}{1.1^i} + \sum_{i=t}^{\log n} 2^t \cdot \frac{1}{1.1^i} = O\left(\frac{2^t}{1.1^t}\right) \eeq

The same asymptotic bound on the expectation   holds also if we condition on the event that the path on which we fell at a certain round did not contain a claw. 

Let $X_j$ denote the number of steps taken by the algorithm in the $j$-th attempt, and let $E_j$ denote the event that a claw is found in attempt $j$. The above discussion implies the following:
\begin{claim}
$\E\left[X_j | \neg E_1 \wedge \neg E_2 \wedge \ldots \wedge \neg E_{j-1}\right] = O\left(\frac{2^t}{1.1^t}\right)$.
\end{claim}

To complete the proof, it suffices to prove the following claim.
\begin{claim} 
\label{claim:prob_claw_found}
$\Pr(E_j | \neg E_1 \wedge \neg E_2 \wedge \ldots \wedge \neg E_{j-1}) = \Theta\left(\frac{2^t}{1.1^t} \cdot \frac{1}{n^{1/10}} \right)$.
\end{claim}
To prove Claim \ref{claim:prob_claw_found}, observe that $E_j$ holds if and only if the random starting point chosen in attempt $j$ belongs to one of the $n^{9/10} / 1.1^t$ paths of length $2^{t}$. There are $2^t \cdot n^{9/10} / 1.1^t$ such points out of a total of $n$ points, and the claim follows by dividing these last two quantities.

Finally, the proof of the lemma follows from these two claims using linearity of expectation and a standard analysis of geometric random variables.
\end{proof}



The main result of this section, given below, is a lower bound showing that algorithms without a certificate require a number of queries that is larger by a multiplicative logarithmic factor compared to the best algorithm with a certificate.

\begin{theorem}
\label{thm:lower_bound_claw_no_cert}
For any algorithm $\A$ (without a certificate), $ \E_{G \leftarrow \Delta}\Query_A(G) = \Omega(n^{1/10} \log(n))$. 
\end{theorem}

To prove Theorem \ref{thm:lower_bound_claw_no_cert}, we first revisit Construction \ref{const:first_const}, discussing an equivalent way to generate the same distribution that is more suitable for our analysis. This alternative construction has some offline components, that take place before the algorithm starts to run, and an online component, that reveals some of the randomness during the operation of the algorithm.

For what follows, let $B(n)$ denote the maximum possible number of (initially isolated) vertices that are added as neighbors of claws in the second part of Construction \ref{const:first_const}. Note that this number is maximized when $t$ takes its minimal possible value, and satisfies $B(n) \leq n^{9/10} / 1.1^{\log (n) / 1000} = n^{9/10 - \Omega(1)}$.

Unlike the original construction, here we think of the vertices of the constructed graph as having one of three colors: black, red, or blue. These colors shall help us keep track of the analysis. In essence (and roughly speaking), blue vertices lead to an immediate victory but they are very rare and unlikely to be found in less than $n^{1/10 + \Omega(1)}$ queries; red vertices are not as rare: finding one of these takes roughly $n^{1/10}$ queries, but $\Omega(\log n)$ red vertices are required to find a claw with constant probability; finally, all other vertices are black, and encountering a black vertex contributes very little to the probability of finding a claw.

\begin{construction}
\label{const:red_blue_black}
We start with an empty graph on $n$ vertices colored black,
and color $B(n)$ of the vertices in blue.

We then construct, as in the first bullet of Construction \ref{const:first_const}, $n / 2.2^{i}$ disjoint paths of length $2^i$ out of black vertices only, for every $\frac{1}{1000}\log n \leq i \leq \frac{1}{100} \log n$.

Next, for every $i$ we pick a subset of $b_i = n^{9/10} / 1.1^i$ paths out of those of length $2^i$. We color the ends of these paths in red.

The last part of the construction happens online, while the algorithm runs. In each time step where the algorithm visits a black vertex, the construction remains unchanged. If the algorithm encounters a red vertex, then we reveal the randomness in the construction in the following way.
\begin{itemize}
    \item Let $I = \{ \frac{1}{1000}\log n \leq i \leq \frac{1}{100} \log n : \text{there exists a path of length $2^i$ with a red end}\}$. Note that initially, $I$ simply contains all values of $i$ in the relevant range; however in the construction $I$ will become smaller with time.
    \item Let $i \in I$ denote the unique integer satisfying that the currently visited red vertex lies on a path of length $2^i$.
    We flip a coin with probability $1/|I|$. If the result is `heads', we consider $i$ as the ``good'' index and do the following: all red ends of paths of length $2^i$ are connected to (isolated) blue vertices, all vertices in the graph are recolored by black, and the construction of the graph is complete.
    \item If the result of the above flip is `tails', we turn all red ends of paths of length $2^i$ to black, and remove $i$ from $I$.
\end{itemize}
Finally, if the algorithm encounters a blue vertex, we pick $i \in I$ uniformly at random to be the ``good'' length, and connect the red ends of paths of length $2^i$ to blue isolated vertices randomly. We then recolor all vertices in the graph to black and consider the construction complete.
\end{construction}

It is straightforward to check that Construction \ref{const:red_blue_black} produces the exact same distribution over graphs as Construction \ref{const:first_const}, and furthermore it does not reuse randomness revealed by the algorithm in previous parts.
Of particular interest is the following observation.
\begin{observation}
Consider any point of time during the online phase of Construction \ref{const:red_blue_black}, and let $I$ be as the defined in the first bullet. Then for any $t \in I$, the probability that $t$ will be the eventual ``good'' index, conditioning on all previous choices made during the construction, is $1/|I|$.
\end{observation}

We say that $\A$ \emph{wins} if it either finds a claw or encounters a blue vertex. The following lemma is a strengthening of Theorem \ref{thm:lower_bound_claw_no_cert}, and its proof immediately yields a proof for the theorem. 
\begin{lemma}
\label{lem:winning}
There exists $C > 0$ such that for any $n \in \N$, any algorithm without a certificate requires at least $C n^{0.1} \log n$ queries to win with success probability $9/10$.
\end{lemma}
From this lemma, it immediately follows that the expected winning time for the algorithm is $\Omega(n^{0.1} \log n)$, which in turn implies the theorem (by definition of winning). Indeed, any algorithm that finds a claw can immediately find a blue vertex and win, since each claw center has two blue neighbors.
Thus, we devote the rest of this section to the proof of Lemma \ref{lem:winning}.

The next (easy) lemma states that encountering a blue vertex is a rare event which will not substantially impact our analysis.
\begin{lemma}
\label{lem:no_blue_vertices}
There exists an absolute constant $\eps > 0$ satisfying the following. For any algorithm $\A$ without a certificate, with high probability $\A$ does not query a blue vertex within $n^{\frac{1}{10} + \eps}$ steps.
\end{lemma}
\begin{proof}
If $\A$ has already won (found a claw), then no blue vertices remain and so the probability to encounter one is zero.

Otherwise, the only chance to encounter a blue vertex at time $t$ is by sampling a vertex from the set of vertices $V_t$ that we did not see so far (i.e., vertices that were not queried and were not revealed to be neighbors of queried vertices). Note that $|V_t| = n - O(t) \geq n/2$ for $t=o(n)$, and the probability that the sampled vertex is blue is bounded by $B(t) / |V_t| \leq 1 / n^{0.1+\Omega(1)}$. The proof follows by a union bound.
\end{proof}

The following result shows that finding $\Omega(\log n)$ red vertices with constant probability requires $\Omega(n^{0.1} \log n)$ queries. To complete the proof, we shall see later that either finding a blue vertex or collecting at least logarithmically many red vertices is essential to win with constant probability.

\begin{lemma}
\label{lem:collision_red_vtxs}
There exists a constant $c > 0$ so that for all $n \in \N$, any algorithm $\A$ which makes $c n^{0.1}\log n$ queries will find less than $\frac{1}{1000}\log n$ red vertices in expectation.
\end{lemma}

\begin{proof}
First, note that we may assume that no blue vertex is ever encountered during the process. If such a vertex is visited, then all red vertices are recolored to black immediately, and so the success probability is zero in this case.

We show that the lower bound applies even if we augment the algorithm with the following additional information revealing some of the randomness in the construction. Clearly, this immediately yields a lower bound for the general case, where the algorithm does not hold such information.

Assume that the algorithm knows in advance, for each vertex $v$ in the graph, whether it is part of a path (i.e., not an initially isolated vertex at the start of the online phase). Further, if $v$ is part of path $I$ of length $2^i$, the algorithm knows $I$. However, the algorithm does not know initially which of the paths have their end vertices colored red.

We denote by $\mathcal{I}_i$ the set of all paths of length $2^i$, and let $V_i = \bigcup_{I \in \mathcal{I}_i} V(I)$ denote the set of all vertices in these paths. Note that $V_i \cap V_j = \emptyset$ for $i \neq j$.
Further let $\mathcal{R}_i \subseteq \mathcal{I}_i$ denote the set of all paths in $\mathcal{I}_i$ whose end vertices are red, and let $\mathcal{B}_i = \mathcal{I}_i \setminus \mathcal{R}_i$ denote the set of all other paths (i.e., all paths with black ends) in $\mathcal{I}_i$. Note that $\mathcal{R}_i$ and $\mathcal{B}_i$ are initially unknown to the algorithm.

The main technical claim of the proof is as follows.
\begin{claim}
\label{claim:single_length}
There exists an absolute constant $\alpha > 0$ satisfying the following for any fixed $\kappa > 0$. For all $n \in \N$, the probability that the algorithm finds a red vertex in $V_i$ within its first $\kappa n^{0.1}$ queries in $V_i$ is bounded by $\alpha \kappa$. This is true regardless of which queries were made outside of $V_i$.
\end{claim}

\newcommand{\ELONG}{{E_{\text{long}}}}
\newcommand{\ESHORT}{{E_{\text{short}}}}

\begin{proof}
Let $E$ denote the event that a red vertex is found within the first $\kappa n^{0.1}$ queries in $V_i$. The statement of the claim is that $\Pr(E) = O(\kappa)$ where the hidden term in the $O(\cdot)$ expression is an absolute constant (independent of $n$).

Consider two events, $\ELONG$ and $\ESHORT$, defined as follows. $\ELONG$ is defined like $E$, with the additional requirement that the red vertex is found in a path in which the algorithm made at least $2^{i-2}$ queries (i.e., at least half of the path length). The complementary event $\ESHORT = E \setminus \ELONG$ is the event that the red vertex is found in a path to which less than $2^{i-2}$ queries were made.

Note, first, that there are more than $n^{0.9}$ paths of length $2^i$, and that we are interested here in the domain where at most $O(n^{0.1})$ queries are being made within $V_i$. Thus, at any time along the process, the probability that a certain path $I \in \mathcal{I}_i$ whose endpoints were not visited yet satisfies $I \in \mathcal{R}_i$ (i.e., has red endpoints) is at most $(1+o(1)) p$, where 
$$p = \frac{n^{9/10} / 1.1^i}{n / 2.2^i} = \frac{2^i}{n^{0.1}}$$ 
was the a priori probability (before the process started) that $I \in \mathcal{R}_i$.

Since $E = \ELONG \cup \ESHORT$, to prove the claim it suffices to show separately that $\Pr(\ELONG) = O(\kappa)$ and $\Pr(\ESHORT) = O(\kappa)$. 

We first bound $\Pr(\ELONG)$. Let $\mathcal{L}_i$ denote the set of paths $I \in \mathcal{I}_i$ satisfying that (i) at least $2^{i-2}$ of the vertices in $I$ were visited during the first $\kappa n^{0.1}$ queries, and (ii) one of the endpoints of $I$ was visited during that time. Note that
$$|\mathcal{L}_i| \leq \frac{\kappa n^{0.1}}  {2^{i-2}}.$$ 

By the above, for each $I \in \mathcal{L}_i$, the probability that $I \in \mathcal{R}_i$ (even conditioned on all previous queries made by the algorithm) is bounded by $(1+o(1)) p$. Taking a union bound, we have that
$$
\Pr(\ELONG) \leq \sum_{I \in \mathcal{L}_i} \Pr(I \in \mathcal{R}_i) \leq |\mathcal{L}_i| \cdot (1+o(1)) p  \leq (1+o(1)) \cdot \frac{\kappa n^{0.1}}  {2^{i-2}} \cdot \frac{2^i}{n^{0.1}} = O(\kappa).
$$
To bound $\Pr(\ESHORT)$, consider any $I \in \mathcal{I}_i \setminus \mathcal{L}_i$ where none of the endpoints of $I$ have been revealed (if an endpoint was revealed to be black then $I \notin \mathcal{R}_i$ and the probability to find a red vertex in $I$ is zero). 
Let $Q \subset I$ denote the set of already queried vertices in $I$. Recall that $|Q| < 2^{i-2}$. 
We claim that conditioned on all the information visible to the algorithm at this point, any vertex in $I \setminus Q$ (including all neighbors of $Q$) has probability at most $\frac{4}{2^{i}}$ to be an endpoint of $I$. 

The proof of this claim relies on a symmetry-based perspective. Instead of thinking of $I$ as a path on $2^{i}$ vertices, one can view it as a cycle on $2^{i}$ nodes where all edges are initially considered unseen. When one of the endpoints of an unseen edge $e$ is queried, we flip a coin with probability $1/r$, where $r$ is the number of unseen edges at this point, and if the result is `heads' we break the cycle into a path by removing $e$ from the cycle. Now to prove our claim, as long as less than $2^{i-2}$ vertices in $I$ were queried, the number of unseen edges is bigger than $2^{i-1}$, and so the probability of each specific node to belong to the edge that will break the cycle is less than $2 / 2^{i-1} = 4 / 2^i$.

Thus, the (conditional) probability that any particular query in a path $I \in \mathcal{I}_i \setminus \mathcal{L}_i$ will reveal a red vertex is at most
$$
\frac{4}{2^i} \cdot (1+o(1)) \cdot \frac{2^i}{n^{0.1}} = O \left(\frac{1}{n^{0.1}} \right).
$$
This is true for any of our $\kappa n^{0.1}$ queries in $V_i$. By a union bound, $\Pr(\ESHORT) = O(\kappa)$.
\end{proof}

The proof of the lemma given Claim \ref{claim:single_length} follows from a simple linearity of expectation argument.
For each $\frac{1}{1000}  \log n \leq i \leq \frac{1}{100} \log n$, Let $I_i$ denote the indicator random variable of whether a red vertex was found in $V_i$, and let $q_i$ denote the number of queries made in $V_i$ in the first $cn^{0.1} \log n$ rounds of the process. It follows from the claim that $\E[I_i]  \leq \alpha q_i / n^{0.1}$, and so
$$
\E[\text{\# red vertices encountered}] \leq \sum_{\frac{1}{1000}  \log n \leq i \leq \frac{1}{100} \log n} \alpha \frac{q_i}{n^{0.1}} \leq \frac{\alpha}{n^{0.1}} cn^{0.1} \log n \leq \alpha c \log n,
$$
which completes the proof by taking $c$ small enough as a function of $\alpha$.
\end{proof}
We now have all the ingredients to complete the proof of Lemma \ref{lem:winning}.

\begin{proof}[Proof of Lemma \ref{lem:winning}]
The algorithm (without certificate) wins if and only if it either encounters a blue vertex, or it encounters a red vertex and the subsequent coin flip results in `heads'. By Lemmas \ref{lem:no_blue_vertices} and \ref{lem:collision_red_vtxs} and a union bound, there exists a small enough absolute constant $C >0$ so that the probability that either of these events happen in the first $C n^{1/10} \log n$ rounds is bounded by $9/10$. The proof follows.
\end{proof}

\subsection{From Claws to Collisions \label{sec:clawstocollision}}
We now briefly define a similar construction aimed at showing the $\Theta(\log n)$ separation for collisions in functions. The same construction and proof also apply if instead of a collision we consider  an ``extended collision'' $H_{a,b,c}$ as defined in Theorem \ref{thm:3.2}.

\begin{construction}
\label{const:collision_construction}
Consider the following process for generating a function $f \colon [n] \to [n]$:
\begin{itemize}
    \item For each integer $\frac{1}{1000}  \log n \leq i \leq \frac{1}{100} \log n$, add $a_i = n / 2.2^i$ disjoint paths of length $2^i$ in the function. 
    
    \item Pick a uniformly random integer $\frac{1}{1000} \log n \leq t \leq \frac{1}{100} \log n$. Pick a collection $\P_t$ of $b_t = n^{9/10} / 1.1^t$ of the paths of length $2^t$. For each such path $P$, let $u$ and $v$ denote the first and last element in the path; set $f(v)$ to be an arbitrary value in $P \setminus \{u,v\}$, which creates a collision. Close all other paths (of all lengths) into cycles: specifically, using the same notation, set $f(v)=u$.
    
    \item For all values $x \in [n]$ that do not participate in any of the above paths, set $f(x) = x$ to be a fixed point.\footnote{We note that the construction can be easily modified to not include fixed points: simply use the remaining values to close cycles of length $2$ or $3$ instead of fixed points, which are essentially cycles of length $1$.}
\end{itemize}
\end{construction}

As is the case with claws in graphs, an algorithm holding a certificate would know the value of $t$, and make walks of length $2^t$ until finding a collision, with a query complexity of $O(n^{1/10})$. Meanwhile, to show the lower bound for an algorithm without a certificate, we use a coloring scheme with only two colors -- red and black -- where elements that are ends of paths which serve as ``candidates'' for a collision are marked red, and all other elements are marked black. Similarly to the above, in order to find a collision with constant probability, the algorithm needs to find $\Omega(\log n)$ red elements with constant probability, which requires $\Omega(n^{1/10} \log n)$ queries.





\remove{
\begin{lemma}
Assuming the total number of queries made thus far is $o(n^c \log(n))$ (or similar), an algorithm whp won't land twice on the same path. 
\end{lemma}

\begin{lemma}
Assuming the total number of queries made thus far is $o(n^c \log(n))$ (or similar), then learning with (exponentially) high probability whether the correct length is $t$, requires $\omega(n^c)$ queries, all spent on walking path of length ``close" to $t$.  

This is independent of $t$, 
\end{lemma}

To be able to use the above lemma I think we need something stronger. I.e claim that an algorithm that spends $O(\frac{n^c}{\log(n)})$ queries walking paths of length $t$ doesn't learn much (i.e the probability that a path is good is very close to $\frac{1}{\log n}$. I think here is where the $\loglog(n)$ might come in -- an algorithm can learn something meaningful in $O(\frac{n^c}{\log\log(n)})$
}

\section{Function Properties Far from Instance Optimal} \label{sec:4}
In this section we study the instance optimality of properties of functions $f \colon [n] \to [n]$. As shown in the previous section, the property of containing a collision is $\Theta(\log n)$-far from instance optimal. We show that any pattern with either at least two collisions or at least one fixed point is \emph{polynomially} far from instance optimality. This is summarized in the theorem below.

\begin{theorem}
\label{thm:function_charac}
Let $H$ be any constant-size oriented graph (possibly with self-edges) where each node has out-degree at most one. Suppose further that $H$ either contains (i) a fixed point (i.e., an edge from a node to itself) or (ii) at least two nodes with in-degree at least two, or (iii) at least one node with in-degree at least three. 

The function property $\P_H$ of containing $H$ as a substructure is %
$\tilde\Omega(n^{1/4})$-far from being instance optimal.
\end{theorem}




The rest of this section is devoted to the proof of the theorem. 
We start with the construction used to prove the theorem.
\begin{construction}
Let $H$ be an oriented graph satisfying the conditions of Theorem \ref{thm:function_charac}.
 Define the \emph{entry vertices} of $H$ to be those vertices with in-degree 0 in $H$ (in the special case where $H$ is a single fixed point, define its single vertex as the entry point). Let $T$ be the total number of entry vertices in $H$. 
 
 Define an input distribution $\Delta$ as follows: first, pick $\alpha \frac{n}{\log n}$ vertices uniformly at random and split them into $N = \alpha n^{1/4} / \log(n)$ disjoint cycles of length $n^{3/4}$, for a small absolute constant $\alpha > 0$. Denote these cycles by $C_1,...,C_{N}$. For each cycle $C_i$ we associate a unique prime number, $p_i$, where all $p_i$'s are in the range $(n^{1/4}/4, {n^{1/4}}/2)$ for an appropriate value of $c$. Note that this is possible due to well-known results on the density of prime numbers. We say that all points contained in the union $\bigcup_{i=1}^{N} C_i$ are of \emph{type 1}.

For each cycle $C_i$, we add paths of length $\alpha n^{1/4}$ entering it, where the distance (in $C_i$) between the entry points of every two adjacent paths entering the cycles is exactly $p_i$. Since each cycle $C_i$ has a length of $n^{3/4}$, it has $\Theta(\sqrt{n})$ paths entering it, each of length $n^{1/4}$. We say that all points participating in these paths are of \emph{type 2}.

Lastly, a collection $x_1, \ldots, x_T$ of exactly $T$ points from $\bigcup_{i=1}^{N} C_i$ is picked uniformly at random conditioned on the event that no two of these points come from the same cycle. Denote the latter event by $E$ and note that $\Pr(E) = 1-o(1)$. For each $x_k$, let $C_{i_k}$ denote the cycle containing it, and turn this cycle into a path ending at $x_k$ by removing the outgoing edge from $x_k$. Finally, insert a copy of $H$ using all $x_1, \ldots, x_T$ as entry points.

To complete the function into one that has size $n$, partition all remaining (unused) points into disjoint cycles of length $n^{3/4}$ each.
\end{construction}

Crucially, an algorithm with a certificate knows the indices $i_1, \ldots, i_T$ (and the corresponding primes $p_{i_1}, \ldots, p_{i_T}$). Given these primes, it turns out that the algorithm is able to find the relevant cycles, walk on them until finding the entry points, and building the full $H$-copy, all using $O(n^{3/4})$ queries. In contrast, for an algorithm without the certificate, the entry points are distributed uniformly over the union of all cycles, and thus a lower bound of $\tilde\Omega(n)$ can be shown.

\begin{lemma} \label{fixedpointupper}
$\E_{f \leftarrow \Delta} RAC(\P_H, f) = O(n^{3/4})$.
\end{lemma}

\begin{lemma} \label{fixedpointlower}
For any algorithm $\A$ (without a certificate), $\E_{f \leftarrow \Delta} \Query_A(f) = \Omega(n / \log n)$.
\end{lemma}

\begin{proof}[Proof Of Lemma \ref{fixedpointupper}]
We show the following stronger claim: For any $f \in \Delta$, $
\mathcal{RAC}(\P_H, f) = O(n^{3/4})$. We provide an algorithm to find all $T$ entry points with $\Omega(1)$ success probability using $O(n^{3/4})$ queries. Once all $T$ entry points are found, completing the copy of $H$ is trivial. 
The algorithm does the following for a large enough constant $C$.
\begin{itemize}
    \item \textbf{Phase 1: sampling short paths.} Pick $Cn^{1/2}$ vertices uniformly at random and follow the path stemming from each point for $Cn^{1/4}$ steps.
    \item \textbf{Phase 2: sampling long paths.} Pick $Cn^{1/4}$ points uniformly at random and follow the path stemming from each point for $Cn^{1/2}$ steps.
    \item \textbf{Phase 3: collection intersections and closing cycles.} Consider any walk $W$ generated in step 2, for which there are three walks generated in step 1 which intersect $W$ at points $y_1, y_2, y_3$ (for this matter, the intersection point of two walks is defined as the first point in which they intersect). If the distance between all of these intersection points is a multiple of $p_{i_k}$ for some $k \in [T]$, we follow $W$ as long as possible (until the walk closes a cycle or reaches a fixed point; the latter means, conditioned on the certificate being correct, that an entry point was found). 
\end{itemize} 

First note that $\Omega(n)$ of the points are of type 2. Thus in Phase 1, with high probability at least $\sqrt{n}$ of the points picked will be of type 2. Since each such point is followed for a length of $n^{1/4}$, the algorithm must query a point that is on the intersection of a path and a cycle. Each of the cycles that contain an entry point will with high probability have $n^{1/4}$ points queried by our algorithm on them. Additionally, with constant probability arbitrarily close to one (for $C$ large enough), Phase 2 will pick at least one point in each of the cycles $C_{i_1}, \ldots, C_{i_T}$ containing an entry point, and subsequently make a walk of length $C\sqrt{n}$ on each of these cycles.

We conclude that with probability $\Omega(1)$, the following holds for all cycles $k \in [T]$: the path generated within $C_{i_k}$ in Phase 2 will intersect at least two paths generated in Phase 1. Thus, Phase 3 will ensure that the rest of $C_{i_k}$ will be queried, until reaching the entry point. Given all query points, $H$ will be queried, as needed. Thus, the success probability of a single iteration of the algorithm is $\Omega(1)$.

It remains to bound the expected query complexity of a single iteration of the algorithm. Clearly, Phases 1 and 2 always take $O(n^{3/4})$ queries. Phase 3 takes $O(n^{3/4})$ queries when restricted to cycles that contain an entry point, since there are only $O(1)$ such cycles. What about other cycles?

With high probability, the following holds for all cycles $C'$ not containing an entry point: at most $O(n^{1/4})$ of the short paths from Phase 1 and at most $O(\log^{1.1} n)$ of the long paths from Phase 2 will intersect $C'$. It follows that for with high probability, simultaneously for all such cycles, the number of short-long intersections is $O(\log^2 n)$. Condition on this high probability event.

Let $C'$ be any cycle not containing an entry point, and let $p'$ be the unique prime associated with this cycle. Also let $p \in \{ p_{i_1}, \ldots, p_{i_k}\}$ be any prime associated with one of the entry point cycles, and let $y_1, \ldots, y_m$ denote all intersection points found in $C'$. Note that all distance between $y_i$ and $y_j$ are divisible by $p'$, but the probability that each such distance is divisible by $p$ is $O(1/p) = O(1/n^{1/4})$. Thus, the probability that there are three intersection points $y_i, y_j, y_l$ whose distances are all divisible by $p$ is $O(\log^{4} n / \sqrt{n})$. Phase 3 will query the whole cycle $C'$ only if this event holds.

Finally, taking a union bound over all values of $p \in \{ p_{i_1}, \ldots, p_{i_k}\}$ and all possible cycles $C'$, we conclude that the expected number of queries made in Phase 3 on incorrect cycles $C'$ (which do not contain entry points) is of order $O(n^{3/4} \cdot \log^4 n / n^{1/4}) = o(n^{3/4})$. 
%

We have seen that a single iteration of the algorithm above succeeds with $\Omega(1)$ probability and has expected query complexity $O(n^{3/4})$. The proof follows by linearity of expectation and standard bounds on geometric random variables.
\end{proof}

\begin{proof}[Proof of Lemma \ref{fixedpointlower}]
Consider the construction right after all type $1$ points are determined, and suppose that the algorithm is given the following information ``for free'': for each cycle
$C_i$, the algorithm knows exactly which points belong to $C_i$ as well as their order along the cycle.
Note that the algorithm does not know a priori the indices $i_1,...,i_T$ of the cycles which will be turned into a path (which ends in an entry point to the $H$-copy). It also does not receive any additional information about type 2 points, aside from the information detailed above.


We say that a point in $[n]$ is \emph{critical} if it participates in the $H$-copy. Note that an algorithm must either query a critical point or make $\Omega(n)$ total queries in order to be correct with constant probability. Thus we bound the number of queries an algorithm makes until it queries a critical point. Once such a point is queried, the algorithm stops running.

Let $C = \bigcup_{i=1}^{N} C_i$ denote the set of all points lying in cycles, and let $B = [n] \setminus C$ denote the set of all other points. With some abuse of notation let $H$ be the (unknown) set of critical points. We note that for any $x \in C$, the probability that $x \in H$ is completely independent of the queries that the algorithm makes in $B$. The same is true for $x \in B$ with respect to queries made in $C$.

Consider first queries that the algorithm makes in $B$. Because of the above independence, the probability that after $i$ queries were made in $B$, the next query will reveal a critical point is bounded by $|H| / (|B|-i) = \Theta(1/n)$. Thus, any algorithm that makes $o(n)$ queries has $o(1)$ probability to query a critical point in $B$. 

\newcommand{\ALONG}{{A_{\text{long}}}}
\newcommand{\ASHORT}{{A_{\text{short}}}}

To analyze the situation in $C$, observe that the event $E$ defined during the construction satisfies $\Pr(E) = 1-1/\Theta(n^{1/4})$. Thus from elementary probabilistic considerations it will suffice to prove the following claim: let $x_1, \ldots, x_k$ be a uniformly random collection of $k$ disjoint points in $C$ without the condition that they come from different cycles (i.e., without conditioning on the event $E$). Then the expected number of queries required for any algorithm (without a certificate) to query at least one of these points is $\Omega(n/\log n)$. However, the proof of this claim is essentially trivial, following from standard estimates for sampling without replacement: in our setting, we have $|C| = \Theta(n / \log n)$ red balls and $T$ green balls in an urn, and the quantity we need to output is the expected amount of time until a green ball is sampled (without replacement). This quantity is of order $\Theta(n / \log n)$.
\end{proof}



\section{Instance Optimality in Graphs} \label{sec:5}


We consider whether questions of the type: ``Does $G$ contains a subgraph $H$" are instance optimal.
 We begin by defining a \kstar{k}.
 
 \begin{definition} [\kstar{k}]
The \emph{\kstar{k}} $S_k$ is a graph with $k+1$ vertices: $k$ vertices of degree 1, all connected to a vertex of degree $k$.
 \end{definition}
 
  For example, a \kstar{1} is a graph consisting of 2 vertices, connected by an edge. A \kstar{2} (or a wedge) contains 2 vertices of degree 1 connected to a third vertex of degree 2. A \kstar{3} is a claw. 


 
 In Section \ref{sec:3}, we have seen that claws are not instance optimal, by showing a $\Omega(\log n)$-separation between an algorithm with a certificate and one without a certificate in some case. What about other choices of $H$? It turns out that if $H$ is a \kstar{1} or \kstar{2} then finding $H$ is instance optimal. If $H$ is any other graph, then finding $H$ is polynomially far from instance optimal. 



\begin{theorem}[\Cref{thm:graphconstantcharacterization} repeated]
Let $H$ be any fixed graph and consider the property $\mathcal{P}_H$ of containing a copy of $H$ as a subgraph. Then $\mathcal{P}_H$ is:
\begin{itemize}
\item instance optimal if $H$ is an edge or a wedge (path with two edges);
\item $n^{\Omega(1)}$-far from instance optimal if $H$ is any graph other than an edge, a wedge, or a claw; and
\item $\Omega(\log(n))$-far from instance optimal when $H$ is a claw.
  \end{itemize}
\end{theorem}


The theorem follows from the lemmas below, together with the results of Section \ref{sec:3}. \Cref{lemma:deg1} proves the first item of the theorem. 
\Cref{lemma:star} proves that for every $H$ that is not a \kstar{k}, finding $H$ is not instance optimal. 
\Cref{lemma:paths} proves the separation for $k$-stars when $k \geq 4$.

 \begin{lemma} \label{lemma:deg1}
$\mathcal{P}_H$ is instance optimal when $H$ is a \kstar{1} (edge) or a \kstar{2} (wedge).
 \end{lemma}
 

 \begin{lemma} \label{lemma:star}
 For all graphs $H$ that are not a \kstar{k}, there exists a distribution, $\Delta$, such that $\E_{G \leftarrow \Delta} \Query_A(G) = \Omega(n)$ for any algorithm $A$ (without a certificate) determining membership in $\P_H$, whereas $\RACC(\P_H,\Delta) = O(n^{1/2} \log n)$. 
 \end{lemma}
 
 \begin{lemma} \label{lemma:paths}
Let $H = S_k$ for $k \geq 4$. There exists a distribution $\Delta$ such that for any algorithm $A$ (without a certificate) determining membership in $\P_H$, $\E_{G \leftarrow \Delta} \Query_A(G) = \Omega(n)$, while $\RACC(\P_H,\Delta) = O(n^{1/2})$. 
 \end{lemma}

\begin{proof}[Proof Of \Cref{lemma:deg1}]
When $H$ is a single edge, it is immediate that the algorithm which repeatedly picks a random vertex and checks whether it has neighbors is instance optimal. When $H$ is a wedge, consider the following algorithm: pick a random vertex $u$. If $u$ has two neighbors, then a wedge was found. If $u$ has one neighbor, $v$, then check if $v$ has a neighbor. It is not hard to verify that the algorithm is instance optimal.
\end{proof}

\begin{proof}[Proof of \Cref{lemma:star}]
For a given graph $H$ we define an input distribution $\Delta$ as follows. Pick $\sqrt{n}$ vertices arbitrarily. These vertices are the center of disjoint $k$-stars. Each remaining vertex is connected to exactly one of the star centers such that each star has a unique degree in $[\frac{\sqrt{n}}{4}, \frac{3 \sqrt{n}}{2}]$. Note that in this construction each vertex is either a star center, or a vertex of degree 1. Next we pick $|H|$ degree-$1$ vertices uniformly at random, such that no two vertices that are chosen belong to the same star. These vertices are then connected to form a clique, and thus in particular contain $H$. Denote this set of vertices by $S$.


Our result follows from the following two claims:
\begin{claim} \label{claim:star:upperbound}
$\RACC(\P_H, \Delta) = O(n^{1/2} \log n)$.
\end{claim}

\begin{claim} \label{claim:star:lowerbound}
For any algorithm $A$ (without a certificate) for $\P_H$, $\E_{G \leftarrow \Delta} \Query_A(G) = \Omega(n)$.
\end{claim}

\begin{proof} [Proof Of \Cref{claim:star:upperbound}]

The algorithm starts by finding the center of each star and its degree as follows:
\begin{algorithmic}

\State $A \leftarrow \emptyset$
\For{$i = 1,...,\Theta(\sqrt{n} \log n)$}
\State Pick an arbitrary vertex $v \in V$, and query the degree of $v$.
\If{ $\text{deg}(v) = 1$}
\State Let $u$ be the neighbor of $v$. $u$ is a center of a star. If $u \notin A$, add $u$ to $A$.
\EndIf
\EndFor
\end{algorithmic}
The probability of any star center to be found in any particular step is $\Theta(1/\sqrt{n})$, so after $\Theta(\sqrt{n} \log n)$ rounds of the above process, $A$ will contain all star centers.

Next, we query the degrees of all star centers. This takes $O(\sqrt{n})$ queries.

Recall that the degree of each star is unique, and thus the algorithm knows from the unlabeled certificate the degrees of the star centers that have a neighbor in $S$; call these star centers ``good''. Finally, we query all neighbors of all good star centers at a cost of $O(\sqrt{n})$ queries. This reveals the clique containing $H$.
\end{proof}

\begin{proof} [Proof of \Cref{claim:star:lowerbound}]
The proof is very similar to that of Lemma \ref{fixedpointlower} and we sketch it briefly.

 Suppose the algorithm is given for free all the stars, but is \emph{not} given $S$. In particular, for every star center $v$, the algorithm knows the set of all neighbors of $v$. The algorithm is required to query at least one node in $S$.
 
 Suppose that we pick $|H|$ vertices uniformly at random among all degree-$1$ vertices. Let $E$ denote the event that no two of them belong to the same star. Then $\Pr(E) = 1 - O(1/\sqrt{n})$. Similarly to Lemma \ref{fixedpointlower}, it suffices to prove the following: if we choose a set $S$ of $|H|$ points from all degree-$1$ vertices without conditioning on $E$, then the expected number of queries required to find one vertex in $S$ is $\Omega(n)$. As in Lemma \ref{fixedpointlower}, this follows from a standard analysis of a sampling with replacement setting where there are $O(1)$ green balls and $O(n)$ red balls, and we are interested in the number of balls that one needs to sample in expectation to get one green ball. 
 %
\end{proof}

The proof of the lemma follows.
\end{proof}


\begin{proof}[Proof of \Cref{lemma:paths}]
We describe the construction for the property of containing a \kstar{k} for $k \ge 4$. Define an input distribution $\Delta$ as follows. Pick $\sqrt{n}$ vertices randomly, and connect them to form a path. Next randomly split the remaining vertices into $\sqrt{n}$ disjoint subsets each of size $\sqrt{n} - 1$ or $\sqrt{n}-2$, $P_1,...,P_{\sqrt{n}}$. Order the vertices in each $P_i$ randomly, and connect them in this order to form a path. Lastly, connect a new vertex $v_i$ to the first vertex in $P_i$ for each $i$, and connect another vertex $v_0$ to $v_1$. See illustration in Figure \ref{fig:paths11}.

\begin{figure}
\begin{center}
\begin{tikzpicture}[scale=0.16]
\tikzstyle{every node}+=[inner sep=0pt]
\draw [black] (12.8,-3.7) circle (3);
\draw (12.8,-3.7) node {$v_0$};
\draw [black] (12.8,-15.1) circle (3);
\draw (12.8,-15.1) node {$v_1$};
\draw [black] (61.1,-15.1) circle (3);
\draw [black] (12.8,-26.5) circle (3);
\draw (12.8,-26.5) node {$v_2$};
\draw [black] (61.1,-26.5) circle (3);
\draw [black] (12.8,-49) circle (3);
\draw (12.8,-49) node {$v_{\sqrt{n}}$};
\draw [black] (61.1,-49) circle (3);
\draw [black] (12.8,-12.1) -- (12.8,-6.7);
\draw [black] (12.8,-18.1) -- (12.8,-23.5);
\draw [black] (15.8,-15.1) -- (58.1,-15.1);
\draw (36.95,-15.6) node [below] {$Path\mbox{ }of\mbox{ }length\mbox{ }\sqrt{n}$};
\draw [black] (12.8,-23.5) -- (12.8,-18.1);
\draw [black] (58.1,-15.1) -- (15.8,-15.1);
\draw [black] (15.8,-26.5) -- (58.1,-26.5);
\draw (36.95,-26) node [above] {$Path\mbox{ }of\mbox{ }length\mbox{ }\sqrt{n}$};
\draw [black] (12.8,-29.5) -- (12.8,-46);
\draw [black] (12.8,-46) -- (12.8,-29.5);
\draw [black] (15.8,-49) -- (58.1,-49);
\draw (36.95,-49.5) node [below] {$Path\mbox{ }of\mbox{ }length\mbox{ }\sqrt{n}$};
\end{tikzpicture}
 \caption{Construction for \Cref{lemma:paths} \label{fig:paths11}}
\end{center}
\end{figure}


The last phase of the construction is to pick a random vertex $u$. We then connect $u$ to $k$ additional vertices. Denote this set of vertices by $S$. 
As we show, the construction establishes the following:
\begin{claim} \label{claim:path:upperbound}
$\RACC(\P_{S_k}, \Delta) = O(n^{1/2})$.
\end{claim}

\begin{claim} \label{claim:path:lowerbound}
For any algorithm $A$ (without a certificate) for $\P_{S_k}$, $\E_{G \leftarrow \Delta} \Query_A(G) = \Omega(n)$.
\end{claim}

\begin{proof}[Proof of \Cref{claim:path:upperbound}]
The algorithm with a certificate knows the index $k$ for which $P_k$ contains the star center. It first finds $v_0, v_1, \ldots, v_{\sqrt{n}}$ in the following way. Pick a random starting point $v$, and walk until encountering either a degree-$1$ vertex or a degree $3$ (or more) vertex. In the former case, we can walk on the other direction from $v$ until reaching a degree $3$ or more vertex. It is easy to check with $O(1)$ queries if the encountered vertex $u$ is part of a $k$-star; assume henceforth that this is not the case. Thus, $u = v_i$ for some $1 \leq i \leq \sqrt{n}-1$ (note that $v_0$ and $v_{\sqrt{n}}$ have degree $1$ and $2$ respectively, not $3$). From here we can locate all vertices $v_i$ for $2 \leq i \leq \sqrt{n}-1$ (in this order or the reverse order) by checking, for each of its neighbors, whether it is degree $3$ (or degree $1$, which marks that we found $v_0$).

Given $v_0, \ldots, v_{\sqrt{n}}$ and the value $k$ for which the star center is in $P_k$, the algorithm simply proceeds by querying all vertices in $P_k$ (via a walk from $v_k$), including the star center.
\end{proof}

\begin{proof}[Proof of \Cref{claim:path:lowerbound}]
 Suppose the algorithm is given, for free, all the paths, but is not given any information about the star center. 
 %
%
The proof follows immediately since the star center was chosen uniformly at random among all vertices in the graph.
\end{proof}
The proof of the lemma follows from the above two claims.
\end{proof}

\section{Finding Claws Is almost Instance Optimal when the input has low Complexity} \label{sec:AlmostIO}

In this section we prove Theorem \ref{thm:inst_opt_graphs_intro}. That is, we prove Conjecture \ref{conjgraph}, that finding a claw in a graph is $O(\log(n))$ close to being instance optimal in the regime where finding a claw can be done in $O(\sqrt{n/\log n})$ queries by an algorithm with an unlabeled certificate. 
For what follows, let $\P_{S_3} = \P_{S_3}(n)$ denote the property of containing a claw (a $3$-star) in $n$-vertex graphs. 

\begin{theorem}[Near-instance optimality in the low-complexity regime] 
\label{thm:near_instance_optimality}
There exist universal constants $C, \alpha > 0$ and a Las Vegas algorithm $A_{\text{all-scales}}$ satisfying the following.
For every graph $G$ on $n$ vertices in which the unlabeled certificate complexity satisfies
$$\RACC(\P_{S_3}, G) \leq \alpha \sqrt{\frac{n}{\log n}},$$
the algorithm $A_{\text{all-scales}}$ detects a claw with expected query complexity at most $C \log n \cdot \RACC(\P_{S_3}, G)$.
\end{theorem}

The proof consists of three main parts, each contained in a separate subsection. \Cref{cleaning} modifies the input graph to be more symmetrical to make all further analysis easier. \Cref{sec:nomerge} shows that it is not possible to find certain type of intersections (merges) between walks in time less than $\sqrt{n/\log(n)}$ unless one finds a claw. Lastly, in \Cref{sec:final} we prove that conditioned on no two paths merging, a ``memoryless'' strategy that follows paths at logarithmically many different paths in parallel, where each path $i$ is followed for length $2^i$,
is $O(\log n)$-instance optimal. 

\subsection{Cleaning up the graph} \label{cleaning}
We start by reducing the problem of finding a claw in an $n$-vertex graph $G$ with $n$ vertices to a closely related problem on a graph $G'$ on the same vertex set, where $G'$ additionally satisfies the following symmetry property, which will be useful for our analysis.

\begin{definition}[path-symmetric graph]
A graph $G$ is \emph{path-symmetric} if for every vertex $v$ of degree $1$ in $G$, the other endpoint $u$ of the simple path $P$ containing $v$ in $G$ is also of degree $1$.
\end{definition}

\paragraph{The construction.} Given an undirected graph $G = (V,E)$ with $|V| = n$, $V$ into disjoint sets $V_0 \cup V_1 \cup V_2 \cup V_{3+}$, where for $i=0,1,2$, $V_i$ contains all degree-$i$ vertices in the graph, and $V_{3+}$ includes all other vertices. 

To each $v \in V_1 \cup V_2$ we can uniquely assign a path $P = P(v)$ that contains $v$.
For any $v \in V_1$ (i.e., endpoint of the path $P(v)$), we say that $v$ is a ``yes'' endpoint if the other endpoint $u$ of $P(v)$ is in $V_{3+}$ (i.e., $u$ is a claw center); otherwise, $v$ is a 
``no'' endpoint (and in this case, $u \in V_1$ as well). Let $V_1^{\text{yes}}$ and $V_1^{\text{no}}$ denote the set of ``yes'' and ``no'' endpoints, respectively.

The transformation $G \mapsto G'$ is defined as follows. 
We add two new vertices $w,w'$ and connect them to all vertices in $V_1^{\text{yes}}$ within $G$. If $|V_1^{\text{yes}}| = 1$, we also connect $w$ to $w'$.
Thus, all vertices from $V_1^{\text{yes}}$ are claw centers in $G'$.
All other vertices in $G$ remain unchanged. The degree of $w,w'$ in $G'$ depends on $n_{\text{yes}} = |V_1^{\text{yes}}|$; it is zero if $n_{\text{yes}} = 0$, two if $n_{\text{yes}} \in \{1,2\}$, and three or more otherwise.
Note that this transformation is deterministic. Thus, an algorithm holding an unlabeled certificate of $G$ can easily construct from it an unlabeled certificate of $G'$.

It immediately follows from the construction that $G'$ is path-symmetric. Our main lemma in this subsection states that finding claws in $G'$ has roughly the same query complexity as in $G$. 
\begin{lemma}
\label{lem:cleaning_up_graph_stars}
$
\RACC(\P_{S_3}, G) = \Theta(\RACC(\P_{S_3}, G')),
$
that is, the expected query complexity of finding a $S_3$ in $G$ with an unlabeled instance is equal, up to a multiplicative constant, to the same query complexity for $G'$.

The same equivalence is true also without access to an unlabeled certificate.
\end{lemma}

The importance of Lemma~\ref{lem:cleaning_up_graph_stars} is that it allows us to eliminate from a graph all vertices of type $V_1^{\text{yes}}$. These vertices are problematic for our analysis; our proof in subsequent subsections relies heavily on symmetry arguments that break down when visiting a vertex of this type. The lemma implies that we can consider graphs in which each path $P$ is of one of two types: either both endpoints of $P$ are claw centers, or both are degree-$1$ vertices.

We note that Lemma~\ref{lem:cleaning_up_graph_stars} holds in any regime, not just the ``low-complexity'' one. In particular, the lemma may be useful toward proving Conjecture~\ref{conjgraph} in the general case.

\begin{proof}[Proof of Lemma~\ref{lem:cleaning_up_graph_stars}]
The proof relies on coupling arguments. 
Given an algorithm $A$ for $\P_{S_3}$ in $G$, consider the following algorithm $A'$ for the same task in $G'$. We simulate the algorithm $A$ with the same set of random bits it would have used on $G$, with the following slight modification. Each time that a ``new'' vertex\footnote{In what follows, a ``new'' vertex is one that did not appear in the output of past queries.} is queried by the simulated copy of $A$, the modified algorithm $A'$ checks if $A'$ is a claw center, or has a neighbor which is a claw center. If not, $A'$ continues the simulation of $A$ (and otherwise, a claw is found).

Conditioned on $A'$ never querying $w$ and $w'$ as ``new vertices'', its query complexity is at most a constant times the query complexity of $A$ on the corresponding instance (removing $w$ and $w'$ and all edges connected to them). If $A'$ queries $w$ or $w'$ at any point, then it wins within $O(1)$ queries. Thus, $\Query_A'(G') = O(\Query_A(G)$.

The other direction is more subtle, and relies in part on symmetry considerations. 
Let $A'$ be an algorithm for $\P_{S_3}$ in graphs on $n+2$ vertices, where we may assume that each time that $A'$ queries a new vertex $v$, it checks whether $v$ is a claw center (and declares a win in this case), and retains the list of neighbors of $v$ if not. This assumption only increases the query complexity by a constant multiplicative factor. We may also assume that each time $A'$ picks a new vertex in the graph, this vertex is picked uniformly at random among all previously unseen vertices.

We consider an algorithm $A$ that operates on an $n$-vertex graph $G$ by simulating $A'$ with the following adaptation. Each time that $A'$ should query a new (previously unseen) vertex, $A$ flips a coin with probability $2/t$, where $t$ is the number of unseen vertices at the current point. If the coin outputs ``heads'', the algorithm discards its entire state and restarts its simulation of $A'$ from scratch (note that previous queries made by $A$ are still counted). 

Let $T = \text{Queries}_{A'}(G')$ denote the expected number of queries of $A'$ on the graph $G'$ created from the input $G$ via the transformation. By Markov's inequality, with probability at least $1/2$ the number of queries until $A'$ finds a claw in $G'$ is at most $2T$. We may assume that, say, $T \leq n/40$, otherwise the statement of the lemma is trivial. In what follows, we claim that $A$ finds a claw in $G$ within $2T$ queries with probability at least $1/5$. The strategy of iteratively running independent copies of $A$ for $2T$ steps each, until a claw is found, has expected query complexity $O(T)$. 
This follows from standard properties of the geometric distribution. 

It remains to show that a single run of $A$ finds a claw with probability at least $1/5$. First, the probability that $A'$ queries $w$ or $w'$ before finding a claw (within the first $2T$ steps) is bounded by $\frac{2}{n} \cdot 2T \leq \frac{1}{10}$. Thus, $A'$ wins without querying any of $w$ or $w'$ with probability $p \geq 2/5$ in $2T$ queries. Condition henceforth on $w$ and $w'$ not being queried by $A'$ in a single simulation. Now, fix any $v \in V_1^{\text{yes}}$ in $G$, that was transformed into a claw center in $G'$. Let $v'$ be the other endpoint of the path $P(v)$ (in $G$), and recall that $v' \in V_{3+}$. By symmetry, the probability that $v$ is queried by the simulated copy of $A'$ before $v'$ is at most $1/2$. Indeed, there are several cases:
\begin{itemize}
\item Case 1: one of $v$ or $v'$ is queried as a ``new vertex'' before being reached by any existing walk of the algorithm. Conditioning on this event, by symmetry, the probability that $v$ is queried before $v'$ is exactly $1/2$.
\item Case 2: One of $v$ or $v'$ is queried as a neighbor of node from within $P$. Again by symmetry, conditioning on this event, the probability of each vertex to be first queried is exactly $1/2$.
\item Case 3: One of $v$ or $v'$ is queried as a neighbor of a vertex outside $P$. Since $v'$ has neighbors outside $P$ and $v$ does not, conditioning on this event, $v'$ is the first to be queried with probability $1$.
\end{itemize}



Summing over all vertices $v \in V_1^{\text{yes}}$, we conclude that the probability that $A'$ found a claw by querying one of the vertices $v \in V_1^{\text{yes}}$ is at most $p/2$. In this case, $A$ does not find a claw. The probability of the complementary event, in which $A$ does find a claw, is at least $p/2 \geq 1/5$. The proof follows.
\end{proof}

\subsection{Hardness of merging walks} \label{sec:nomerge}

In this subsection we prove a hardness result concerning \emph{merging walks} in graphs. We shall care about a property of \emph{walks}, as \emph{oriented subgraphs} of the input (undirected) graph. Although the main focus of this subsection is on the graph case, the definition below of a walk graph is suitable in both graphs and functions.

\begin{definition}[Algorithm's walk graph]
\label{def:walk-graph}
The walk graph of a query-based algorithm $A$ in $n$-vertex graphs (in the query model) and functions $f \colon [n] \to [n]$ (treated as out directed graphs with out-degree $1$) is initialized as an empty graph on $[n]$. 

Each time that an algorithm $A$ queries a vertex $v$ to obtain a neighbor $u$ of it (or an out-neighbor, in the functions case), we add the oriented edge $v \to u$ to $G_A$. 
\end{definition}
A walk is any connected component (in the undirected sense) of the walk graph, preserving the orientations of the walk graph. That is, it is an object of the following structure:
$$
v_{-k} \leftarrow v_{-k+1} \leftarrow \ldots \leftarrow v_0 \rightarrow v_1 \rightarrow \ldots \rightarrow v_l,
$$
where $k,l \geq 0$ (and can equal zero), and $v_0$ is the first vertex to be queried among $\{v_i : -k \leq i \leq l\}$.

Note that the walk graph is kept oriented (each edge only has a single orientation); if $v \to u$ has already been added to $G_A$, then $u \to v$ cannot be added to it in the future.

\begin{definition}[Merging walks]
\label{def:merging_walks}
Consider the walk graph $G_A$ for an algorithm $A$ in a graph $G$. We say that two walks $$v_{-k} \leftarrow \ldots \leftarrow v_{-1} \leftarrow v_0 \to v_1 \to \ldots \to v_l \ \ \text{ and }\ \  u_{-k'} \leftarrow \ldots \leftarrow u_{-1} \leftarrow u_0 \to u_1 \to \ldots \to u_{l'}$$ 
in $G_A$ with distinct starting points $v_0 \neq u_0$ \emph{merge} if their endpoints intersect non-trivially: $$\{v_{-k}, v_l\} \cap \{u_{-k'}, u_{l'}\} \neq \emptyset.$$ 
\end{definition}

Our main result is a universal (not graph-specific) lower bound on the query complexity of finding merging paths, as long as a the walk-graph of the algorithm does not contain a claw center.
\begin{lemma}[Hardness of merging for graphs]\label{lem:hardness_merge_graphs}
There exists $\alpha > 0$ such that the following holds. For
every $n$-vertex graph $G$ and algorithm $A$ making up to $q = \alpha \sqrt{n / \log n}$ queries to $G$, the probability that walks merge in $G_A$ at least once during the run of $A$ and that further $G_A$ does not contain a claw center at the time the merging takes place, is at most $1/5$. 
\end{lemma}
The constant $1/5$ here is arbitrary (and can be replaced with any other small positive constant). 
The above lemma implies that the ``strategy'' of first finding merging walks and then using them to find a claw would require at least $\Omega(\sqrt{n /\log{n}})$ queries. In the next subsection, we show the complementary result: that in a graph where no merges have been observed, a simple strategy that does not require an unlabeled certificate is almost instance-optimal.



\paragraph{Toward the proof of Lemma~\ref{lem:hardness_merge_graphs}.} We start by introducing objects and ideas to be used in the proof of the lemma.
We may assume that the input graph $G$ does not contain claw centers, i.e., has maximum degree at most $2$. (Indeed, querying any edge that contains a claw center makes the algorithm fail, and it is easy to show that removing all these edges cannot decrease the probability to win.)
Therefore, the graph $G$ is a disjoint union of simple paths, simple cycles, and isolated vertices. For what follows, it will be convenient to measure the length of a path/cycle as the number of \emph{vertices} (not edges) it contains, that is, an isolated vertex is a path of length $1$.

We bucket the paths in $G$ into sets of similarly-sized paths and cycles. Let $t = \log_{1.1} n = \Theta(\log n)$, 
and for every $0 \leq i < t$ let $\mathcal{P}_i$ denote the collection of all paths and cycles in $G$ containing at least $1.1^i$ and less than $1.1^{i+1}$ vertices.
Further let $N_i$ denote the sum of lengths of paths and cycles in $\mathcal{P}_i$. Note that these collections are disjoint and thus $\sum_{i=0}^{\log n} N_i = n.$ Define
$$
S = \left\{ 0 \leq i \leq t : N_i \geq \frac{\sqrt{n \log n}}{t}  \right\},
$$
where we note that $\sum_{i \notin S} N_i \leq \sqrt{n \log n}$. In particular, if an algorithm makes $q \leq \alpha \sqrt{n / \log n}$ queries then with probability $1-2\alpha$ none of the vertices in $\bigcup_{i \notin S} \mathcal{P}_i$ will be queried. That is, we may assume that the walk-graph of the algorithm does not intersect $\bigcup_{i \notin S} \P_i$, and contains only walks strictly in paths outside this union.

We prove the statement of Lemma~\ref{lem:hardness_merge_graphs} under an even stronger algorithmic model. Each time that a ``new'' vertex $v$ is being queried by the algorithm $A$, we immediately notify the algorithm about which bucket $\P_i$ the vertex $v$ belongs to. Note that this implies, in particular, that at all times during the algorithm's run, each walk $W$ in the walk-graph of $A$ is contained in a bucket known to the algorithm. We show that even under this stronger assumption, it is hard to merge walks.

Our first main claim is a concentration of measure argument regarding the number of walks within each bucket. It is used in the proof of Lemma~\ref{lem:hardness_merge_graphs} to show that walks in the same bucket have, very roughly speaking, a high ``degree of freedom'', resulting in good bounds on the merging probability.

\begin{claim}[Properties of walks]
\label{claim:walk_properties}
There exists a constant $\alpha > 0$ satisfying the following. Let $A$ be a query-based algorithm in $G$ making $q \leq \alpha \sqrt{n / \log{n}}$ queries. The following statements hold with probability at least $9/10$.
\begin{enumerate}
\item For all $0 \leq i < t$, the number of walks in $G_A$ contained in $\bigcup_{P \in \P_i} P$ is bounded by $\frac{10 N_i q t}{n}$.
\item For all $i \in S$ for which $1.1^i \leq \sqrt{\frac{n}{\log n}}$, the number of such walks is smaller than $\frac{1}{2} |\P_i|$.\footnote{Note that the quantity $|\P_i|$ counts the \emph{number} of paths in $\P_i$, not their total length.} 
\end{enumerate}
\end{claim}
\begin{proof}
For the first bullet, note that the expected number of ``new queries'' falling in $\bigcup_{P \in \P_i}$ is $N_i q / n$. By Markov inequality, the probability that (for a specific $i$) the number of walks is bigger than the expectation by $10t$ is $1/10t$. The proof follows by a union bound since we have $t$ buckets.


We now turn to the second bullet. The condition $1.1^i \leq \sqrt{n / \log n}$ implies
$$
\frac{10 N_i q t}{n} \leq \frac{20 \cdot 1.1^i |\P_i| q t}{n} < \frac{|\P_i|}{2},
$$
where the last inequality holds provided that the constant $\alpha$ is small enough. The proof follows.
\end{proof}

\begin{proof}[Proof of Lemma~\ref{lem:hardness_merge_graphs}]
As discussed above, we can assume that $G$ is a collection of disjoint simple paths (without any claw center vertices). 
We can simulate any (up to) $q$-query algorithm on such a graph $G$ with the following algorithm $A$ making up to $2q$ queries. First $A$ makes $q$ ``new vertex'' queries, by iteratively picking previously unseen nodes in $G$ and querying them. Let $Q$ be the set of all such queried nodes.
In the second step, $A$ simulates the original algorithm: each time that the latter decides to query a new vertex, $A$ picks an unused vertex from $Q$ and provides it to the simulated copy. If the algorithm decides to extend an existing walk, $A$ does so as well. It is easy to see that the probability of $A$ to encounter merging walks in $G$ is at least as that of the original algorithm. The probability of both algorithms to encounter a claw is zero, since $G$ is claw-free. Thus, it suffices to prove the lemma for $A$ as above. 

By a standard birthday paradox argument, with high probability no two walks will merge in the initial querying step (which results in $q$ oriented edges), so long as $q = o(\sqrt{n})$. We condition on this event for the rest of the proof. We also condition on the statement of Claim~\ref{claim:walk_properties}. 

For any $0 \leq i < t$, fix a linear ordering of the vertices in $\P_i$, so that each path in $\P_i$ consists of consecutive vertices in this order. Now, for each walk $W$ in $G_A \cap \P_i$, one can uniquely define the \emph{location} of $W$ as the element of $W$ that is smallest in the linear ordering.

We consider a walk $W \in \P_i$ to be \emph{short} if $|W| \leq \frac{1}{3} \cdot 1.1^i$, and \emph{long} otherwise.
The proof proceeds by bounding two types of probabilities: (i) the probability of a short walk $W$ to merge with another short walk when an edge is added to $W$, and (ii) the probability of a long walk $W$ to reside in the same path as another walk (short or long). 

Let us start by analyzing the first type of event.
\begin{claim}
\label{claim:short-short-merge}
Let $W$ be a short walk in $\P_i$ and suppose that we extend $W$ by an edge $e$. The probability that $W \cup \{e\}$ will merge with one of the other short walks, conditioning on the statement of Claim \ref{claim:walk_properties}, is bounded by $\sqrt{\frac{\log n}{n}}$.
\end{claim}
\begin{proof}
We claim that the probability that $W$ will merge with any specific short walk $W'$ in a single step is bounded by $4/N_i$.

Consider first the case that both $W$ and $W'$ do not contain an endpoint of a path from $P$.
Fix the location of $W$ and all other walks in $\mathcal{P}_i$, except for $W'$. We claim that $W$ and $W'$ merge only if $W'$ is located in one specific location (either one place after the largest element of $W$, or one place before the smallest of $W$, depending on the direction of the walk). On the other hand, conditioning on the event that no merging has happened so far, $W'$ is distributed uniformly among all locations for which $W'$ does not intersect any of the other walks in $\P_i$. We call such locations \emph{free}.

If $i$ satisfies $1.1^i \leq \sqrt{\frac{n}{\log n}}$, then (by the second part of Claim~\ref{claim:walk_properties}) a subset $\P_i^{\text{free}}$ of at least $\frac{1}{2} |\P_i|$ of the paths in $\P_i$ do not contain any walk (except for possibly $W'$). The number of free locations in each $P \in \P_i^{\text{free}}$ is at least $1.1^i / 2$. Thus, at least $N_i / 4$ of the locations in $\bigcup_{P \in \P_i} P$ are free. 
If $1.1^i > \sqrt{n / \log{n}}$, then since $q < \frac{1}{10}\sqrt{\frac{n}{\log n}}$ (provided $\alpha < 1/10)$, at most $N_i / 2$ of the locations in $\P_i$ will be occupied, and more than $N_i / 4$ locations will be free. 
In both cases, the number of free locations for $W'$ is at least $N_i / 4$. 

So far we dealt with the case that $W$ and $W'$ do not contain an endpoint of a path. In the case that both $W$ and $W'$ contain an endpoint, their merging probability is zero. It remains to consider the case where one of $W$ and $W'$ contains an endpoint $v$ in $\P_i$ and the other does not. The case where  $v \in W$ is identical to the analysis above (when $W$ and $W'$ do not contain an endpoint). The case where $v \in W'$ is symmetric, where instead of fixing the location of $W$ and randomizing over the location of $W'$, we fix $W'$ and randomize $W$. The bounds we obtain are identical.

Thus, the probability for merging of $W$ with $W'$ in this case is bounded by $4/N_i$. Taking a union bound over all possible walks $W'$,
the total probability for merging (with at least one short walk $W' \neq W$) is bounded by
$$
\frac{4}{N_i} \cdot \frac{10 N_i q t}{n} = \frac{40 q t} {n} \leq \sqrt{\frac{\log n}{n}},
$$
by the first part of Claim~\ref{claim:walk_properties} and provided that $\alpha$ is small enough.
\end{proof}

The second claim that we need considers the behavior of walks in the first step where they cross the threshold for becoming a long walk.

\begin{claim}
\label{claim:long-short-merge}
Let $W \in \P_i$ be a walk of length at least $\frac{1}{3} \cdot 1.1^i$ residing in a path $P$. The probability that $P$ contains at least one additional short walk is bounded by $1.1^i \sqrt{\frac{\log n}{n}}$.
\end{claim}
\begin{proof}
Let $W' \neq W$ be a short walk in $\P_i$. Fix the location of all walks in $G_A$ except for that of $W$. We start with the case that $W'$ does not contain an endpoint.
As in the previous claim, conditioned on no merging happening so far, the location of $W'$ is distributed uniformly at random within a set of size at least $N_i / 4$. Among these, less than $1.1^i$ of the possible locations are in $P$ (note that $|P| < 1.1^{i+1}$, and that at least $\frac{1}{3} \cdot 1.1^i$ locations in it are occupied by $W$). By a union bound, the probability that $P$ contains at least one such $W'$ is bounded by
$$
\frac{1.1^i}{N_i/4} \cdot \frac{10 N_i q t}{n} = \frac{10 \cdot 1.1^i q t}{n} \leq \frac{1.1^i}{2} \sqrt{\frac{\log n}{n}},
$$
provided $\alpha$ is small enough. The case that $W'$ contains an endpoint of a path in $\P_i$ is treated similarly, except that there are at least $\frac{N_i} {4 \cdot 1.1^i}$ locations (that contain an endpoint) to uniformly choose from for $W'$, with only at most two of the locations belonging to $P$. A similar accounting gives a bound of $\frac{1.1^i}{2} \sqrt{\frac{\log n}{n}}$ for this case as well, provided small enough $\alpha$. The proof follows by a union bound on these two cases.
\end{proof}





To complete the proof of the lemma, recall that the number of events of adding an edge to an existing short walk is bounded by $q \leq \alpha \sqrt{n / \log n}$, and the probability of each such event to lead to merging with another short walk is bounded by $\sqrt{\frac{\log n}{n}}$. By a union bound, the probability for a short-short merge is at most $\alpha$. The other type of merging events happen when at least one of the walks is long. The probability of such merging to happen is bounded by the sum of probabilities, over all long walks $W$ in $G_A$, of $W$ having another $G_A$-walk residing in the same path. By Claim~\ref{claim:long-short-merge}, the latter sum is bounded by 
$$
\sum_{W \in \P_i : W \text{ long}} 1.1^i \sqrt{\frac{\log n}{n}} \leq \sum_{W \in \P_i : W \text{ long}} 3|W| \sqrt{\frac{\log n}{n}} \leq 3q \sqrt{\frac{\log n}{n}} \leq 3\alpha, 
$$
where the second inequality holds because the sum of length of (long) walks in $G_A$ is bounded by twice the total number of queries. The sum of probabilities of all ``bad'' events in the analysis (which may lead to possible merging) is bounded by $\frac{1}{10} + O(\alpha) \leq \frac{1}{5}$ for small enough $\alpha$.
\end{proof}

\subsection{On algorithms for finding claws without merging}\label{sec:final}
We next establish the near-instance-optimality of claw detection in the absence of merging between walks. Our main result is as follows. 

\begin{lemma}[Near instance-optimality in graphs without merges]
\label{lem:instance_optimality_graphs_no_merges}
There exists a constant $\alpha > 0$ for which the following holds.
Let $G$ be a path-symmetric graph admitting an algorithm $A$ for $\P_{S_3}$ with expected query complexity $q_G \leq \frac{1}{20}\sqrt{n}$. Suppose that the probability of $A$ to cause merging in $G_A$ within its first $2q_G$ queries, provided that it has not found a claw before merging, is bounded by $1/5$. 
Then there exists an algorithm $A_{\text{all-scales}}$ without an unlabeled certificate, that finds a claw with expected query complexity $O(q_G \log n)$.
\end{lemma}

\paragraph{The algorithm $A_{\text{all-scales}}$.}
Before proving the lemma, we present the  algorithm $A_{\text{all-scales}}$ operating on a graph $G$ on $n$ vertices $v_1, \ldots, v_n$. The algorithm $A_{\text{all-scales}}$ first picks a random permutation $(v_{\pi(1)}, v_{\pi(2)}, \ldots, v_{\pi(n)})$ of the vertices. 
Now, for each $0 \leq i < \log n$, define $A^{(i)}$ as follows. 
\textit{Iterate} the following for $j=1,2,\ldots,n$: 
\begin{enumerate}
\item Set $v = v_{\pi(j)}$.
\item Walk from $v$ in both directions (alternately) for $2^{i+1}$ steps, or until one of the following is found: a claw, a path endpoint, or a closing of a simple cycle. Terminate algorithm and report success if a claw was found. In the other two cases, stop the iteration and continue to the next one.
\end{enumerate}

The algorithm $A_{\text{all-scales}}$ simulates one copy of each algorithm $A^{(i)}$, $0 \leq i < \log n$, in an interleaved manner as follows. $A_{\text{all-scales}}$ proceeds in rounds, where in each round it simulates exactly one step for each algorithm $A^{(i)}$. (Note that a single round of $A_{\text{all-scales}}$ uses a total of $\log n$ queries, one for each copy.)


\begin{proof}[Proof of Lemma~\ref{lem:instance_optimality_graphs_no_merges}]
Let $A$ be any algorithm making up to $q$ queries to $G = (V,E)$. Without loss of generality, we may assume there is a random string $r(q,G) = (v_1, v_2, \ldots, v_q) \in V^q$, picked uniformly at random among all sequences of $q$ vertices \emph{without repetitions}, and used as follows. Each time that $A$ intends to query a ``new'' (previously unseen) vertex in $V$, it picks the first vertex from this sequence that was not queried until now, and uses it as the new vertex. We may also view each $A^{(i)}$ (when running for a total of $q$ iterations) as using the same type of random string $r(q,G)$, where in the $i$-th time that a new vertex $v \in V \setminus S^{(i)}$ is picked in the first bullet of the algorithm description, we set $v = v_i$.

We define the following events of interest.
\begin{itemize}
\item $E_{\text{no-merge}}(A, q, G, r)$ is the event that when running $A$ for $q$ queries on $G$ with $r$ as the random string, no merging of walks occurs before the first time a claw is found (and, if no claw is found during this entire time window, then simply no merging occurs for the entire window).
\item $E_{\text{skip}}(A,q,G,r)$ is the event that, at some point during the run of $A$ on $G$ (with up to $q$ queries and random string $r$), some vertex $v_i$ in the random string is encountered by a query made by $A$ \emph{before} being selected as a ``new'' vertex. This means that $A$ will skip $v_i$ when picking new vertices in the future. We denote by $E_{\text{no-skips}}(A,q,G,r)$ the complementary event to $E_{\text{skip}}$. The event $E_{\text{no-skips}}$ corresponds to the situation where for all $1 \leq i \leq q$, in the $i$-th occurrence where $A$ picks and queries a ``new'' vertex, this vertex will be $v_i$.
\item $E_{\text{claw}}(i, A, q, G, r)$ is the event in which the following two conditions hold: (i) $A$ detects a claw-center, $w$, in $G$ within up to $q$ rounds using the random string $r$; and (ii) the walk $W$ containing $w$ in the walk graph $G_A$ at the time where $w$ is first queried is of length (number of vertices) between $2^{i}$ and $2^{i+1}-1$.\footnote{We will only be interested in this event in situations where walk merging \emph{did not happen} before a claw was found. Thus, this event is well defined: there is exactly one walk containing $w$ when it is first queried.}
\end{itemize}

Our main technical claim of the proof is that $A^{(i)}$ stochastically dominates $A$ in the following asymptotic sense.
\begin{claim}
\label{claim:optimality_of_Ai}
Fix $A, q, G, r$ as above. Let $E_{\text{no-merge}}$ and $E_{\text{no-skips}}$ be defined as above with respect to $A,q,G,r$. 
For any $0 \leq i < \log n$, 
\begin{equation*}
\label{eqn:claw_stochastic_domination}
     \Pr(E_{\text{claw}}(i, A^{(i)},  4q, G, r) \ \ |\ \  E_{\text{claw}}(i, A, q, G, r) \cap E_{\text{no-merge}} \cap E_{\text{no-skips}}) = 1.
\end{equation*}
\end{claim}

In particular, the proof shows that when $A$ and $A^{(i)}$ are coupled, using the same random string $r$, then
conditioned on both events $E_{\text{no-merge}}$ and $E_{\text{no-skips}}$ holding with respect to the algorithm $A$, if the algorithm $A$ can find a claw within $q$ queries, then $A^{(i)}$ will (deterministically) find the same claw, albeit with a slightly higher query cost of $4q$ (instead of $q$).  

\begin{proof}[Proof of Claim~\ref{claim:optimality_of_Ai}] 

We can view each of $A$ and $A^{(i)}$ as maintaining a collection of walks. Recall that the event $E_{\text{no-merge}}$ holds, that is, the walks in $A$ do not merge (in $A^{(i)}$ we do not care about merging). By symmetry considerations, we may assume that $A$ operates according to the following ``older first'' principle:
if two walks $W$ and $W'$ in $G_{A}$ have the exact same shape at some moment during the run of $A$, and if $A$ decides to extend one of $W$ and $W'$ by an edge at that moment, then the extension will take place at the walk that is older, i.e., that was created earlier. 

Since the event $E_{\text{no-skips}}$ holds, the $j$-th ``new'' vertex to be queried by the algorithm (for all $1 \leq j \leq q$) will be $v_j$. Let $W_j$ denote the walk in $G_A$ that emanates from $v_j$ (if $v_j$ was not yet queried, we think of $W_j$ as a length-$0$ walk). Observe that for two walks $W_a$ and $W_b$ where $a < b$ that have not yet encountered a path endpoint (or a claw) at some point in time, it holds that $|W_a| \geq |W_b|$ at that time. This is a direct consequence of the `older first'' principle (and the no-merging property). 

Suppose that $A$ first encounters a claw center $w$ at time $t\leq q$, as part of a walk $W_j$ of length between $2^i$ and $2^{i+1}-1$. That is, the event $E_{\text{claw}}(i, A, G, t, r)$ holds while $E_{\text{claw}}(i, A, G, t-1, r)$ does not hold.
We claim that $A^{(i)}$ will encounter (with probability $1$) the same claw center within its first $4q$ queries, implying in particular that with probability $1$, the event $E_{\text{claw}}(i, A^{(i)}, G, 4t, r)$ also holds in this case.

Indeed, suppose this is the case for some walk $W_j$ at some time $t$. For any vertex $v \in G$, denote by $d_G(v)$ the distance of $v$ from an endpoint of the path it belongs to; if $v$ belongs to a cycle, then $d_G(v)$ is defined as the length of the cycle.
Let $$\mathcal{I} = \{j' \leq j : \text{$d_G(v_{j'}) \geq 2^i$}\} \ \ \text{ and }\ \  \bar{\mathcal{I}} = [j] \setminus \mathcal{I}.$$

Recall that $|W_j| \geq 2^{i}$ at time $t$. By the ``older first'' principle, $|W_{j'}| \geq 2^i$ holds at time $t$ for any $j' \in \mathcal{I}$, and $|W_{j'}| \geq d_G(v_{j'})$ for any $j' \in \bar{\mathcal{I}}$. Thus, 
$$t \geq 2^i \cdot |\mathcal{I}| + \sum_{j' \in \bar{\mathcal{I}}} d_G(v_{j'})$$. 

On the other hand, by definition of $A^{(i)}$ and since its randomness is coupled with that of $A$ as detailed above, $A^{(i)}$ will spend at most $2 \cdot 2 \cdot 2^{i+1} = 2^{i+2}$ queries for the walk emanating from each $v_{j'}$ with $j' \in \mathcal{I}$; and at most $2 d_G(v_{j'})$ queries for $j' \in \bar{\mathcal{I}}$. Thus, the total number of queries until it encounters the claw center $w$ is at most
$$
\cdot 2^{i+2} |\mathcal{I}| + 2 \sum_{j' \in \bar{\mathcal{I}}} d_G(v_{j'}) \leq 4t,
$$
as desired.
\end{proof}

We now resume the proof of Lemma~\ref{lem:instance_optimality_graphs_no_merges}. 
Let $A$ be an algorithm achieving an expected query complexity $q_G$ for finding a claw in $G$. By Markov's inequality, $A$ finds a claw with probability at least $1/2$ within $2q_G$ queries. Fix a random string $r = r(2q_G,G)$. By the statement of the lemma, 
$\Pr(E_{\text{no-merge}}(A, 2q_G, G, r)) \geq 4/5$. 
Further, it is easy to show that the probability of $E_{\text{skip}}(A,q,G,r)$ for any algorithm $A$ making $q$ queries is bounded by $2q^2 / n$. 
In our case $q = 2q_G \leq \frac{1}{10}\sqrt{n}$, and this probability is smaller than $1/20$. We conclude that $$\Pr(E_{\text{no-merge}} \cap E_{\text{no-skips}}) \geq 3/4.$$

Next, observe that $A$ finds a claw after at most $2q_G$ rounds if and only if at least one of the events $E_{\text{claw}}(i, A, 2q_G, G, r)$ holds. Since the total probability to find a claw is $1/2$, we get that
$$
\Pr\left(\bigcup_{i=0}^{\log n - 1} E_{\text{claw}}(i,A, 2q_G,G,  r)\ \  |\ \  E_{\text{no-merge}} \cap E_{\text{no-skips}}\right) \geq \frac{1}{2} - \frac{1}{4} = \frac{1}{4}.
$$

Thus, by Claim~\ref{eqn:claw_stochastic_domination} we have
\begin{align*}
\Pr\left(\bigcup_{i=0}^{\log n - 1} E_{\text{claw}}(i,A^{(i)}, 8q_G, G,  r) \ \ |\ \  E_{\text{no-merge}} \cap E_{\text{no-skips}}\right) \geq \frac{1}{4}.
\end{align*}
Thus, $A_{\text{all-scales}}$ finds a claw with probability $1/4$ within $8q_G$ rounds (with a total query complexity of $8q_G \log n$). By standard properties of the geometric distribution, $A_{\text{all-scales}}$ finds a claw with expected query complexity $O(q_G \log n)$, and the proof follows.\footnote{Note that the geometric distribution requires independence between different trials, whereas the way we defined $\A_{\text{all-scales}}$ may make it seem at first that there are nontrivial dependencies, due to the algorithm generating a random permutation in its setup. However, the algorithm can be defined equivalently using a ``with replacement'' version without the need to specify a random permutation. Indeed, when starting a walk from a vertex $v$ that was already used as a source, we can simply reuse the old walk from $v$ for free, without paying any extra queries.}
%
\end{proof}

It remains to combine all pieces for the proof of Theorem \ref{thm:near_instance_optimality}.
\begin{proof}
By Lemma~\ref{lem:cleaning_up_graph_stars}, we may assume that the input graph $G$ is path-symmetric. Suppose that $G$ admits an algorithm (with an unlabeled certificate) for $\P_{S_3}$ with expected query complexity at most $q \leq \frac{\alpha}{2} \sqrt{\frac{n}{ \log n}}$, where $\alpha$ is as in Lemma~\ref{lem:hardness_merge_graphs}. By the latter lemma combined with Lemma \ref{lem:instance_optimality_graphs_no_merges}, the algorithm $A_{\text{all-scales}}$ detects a claw in $G$ with expected query complexity $O(q \log n)$, as desired.
\end{proof}


\addcontentsline{toc}{section}{References}
\bibliographystyle{alpha}
\bibliography{main}

\appendix

\end{document}